%% file: main.tex
\pdfoutput=1
\documentclass[aps,prx,twocolumn,superscriptaddress,nofootinbib,longbibliography]{revtex4-2}
\input{glyphtounicode}
\usepackage{amsmath,amssymb}
\usepackage{booktabs}
\usepackage{array}
\usepackage[table]{xcolor}
\usepackage{graphicx}
\usepackage{tikz}
\usetikzlibrary{arrows.meta,positioning,shapes.symbols,shapes.geometric,calc,backgrounds,fit}
\usepackage{amsthm}
\usepackage[hidelinks]{hyperref}
\usepackage[normalem]{ulem}


\renewcommand{\topfraction}{0.92}
\renewcommand{\bottomfraction}{0.8}
\renewcommand{\textfraction}{0.06}
\renewcommand{\floatpagefraction}{0.7}
\theoremstyle{definition}
\newtheorem{definition}{Definition}
\theoremstyle{plain}
\newtheorem{theorem}{Theorem}

\newtheorem{corollary}{Corollary}
\theoremstyle{remark}

\theoremstyle{plain}
\newtheorem*{thmformal}{Theorem~\ref{thm:main} (formal version)}

\newcommand{\ZA}{\langle Z_A\rangle}
\newcommand{\supp}{\mathrm{supp}}
\newcommand{\MMD}{\mathrm{MMD}}

\begin{document}

\title{``Train classical, deploy quantum'' requires rethinking generalization}

\author{Snehal Raj}
\affiliation{QC Ware Corp., Palo Alto, USA and Paris, France}
\affiliation{LIP6, CNRS, Sorbonne Universit\'{e}, Paris, France}
\author{Natansh Mathur}
\affiliation{QC Ware Corp., Palo Alto, USA and Paris, France}
\author{Alejandro Perdomo-Ortiz}
\email{alejandro.perdomoortiz@qcware.com}
\affiliation{QC Ware Corp., Palo Alto, USA and Paris, France}

\begin{abstract}
Generative models have become central across science and industry, from image and text
    synthesis to the design of molecules and materials. Quantum generative models are considered one
    of the most promising applications for quantum computers, since a quantum circuit naturally
    produces samples from the distribution it encodes, and for suitable circuits that distribution is
    believed to be hard for any classical computer to reproduce. A leading strategy trains these
    models on a classical computer and reserves the quantum device for generating samples at
    deployment. This is possible when the training loss can be evaluated on a classical computer. A
    prime example is the maximum mean discrepancy (MMD$^2$), a moment-matching loss that compares the
    model and the data through their Pauli-$Z$ correlations. Research so far has asked whether such
    models can be trained and whether their sampling is hard; whether minimizing such an objective
    yields a model that \emph{generalizes}, rather than one that merely reproduces the training
    statistics, remains poorly understood. We benchmark thirteen quantum and classical generative models by direct sampling on two
    application-inspired datasets: first a cardinality-constrained dataset at up to $30$ qubits and
    second a dataset of genomic single-nucleotide variants, whose valid set is the observed data.
    Models that converge the loss to the same value differ widely in how much of the unseen
    valid set they cover. These results indicate that a converged moment-matching loss is not a
    reliable measure of generalization, and that a train-classical, deploy-quantum workflow has
    to measure generalization by sampling the trained model, a step that at the sizes of interest is
    believed to require the quantum device.
\end{abstract}

\maketitle

\renewcommand{\topfraction}{0.92}
\renewcommand{\bottomfraction}{0.7}
\renewcommand{\textfraction}{0.06}
\renewcommand{\floatpagefraction}{0.72}
\renewcommand{\dbltopfraction}{0.92}
\renewcommand{\dblfloatpagefraction}{0.72}
\setcounter{topnumber}{3}
\setcounter{bottomnumber}{2}
\setcounter{totalnumber}{5}
\setcounter{dbltopnumber}{3}

\section{Introduction}
\label{sec:intro}

\begin{figure*}[t]
\centering
\begin{tikzpicture}[>={Latex[length=2.4mm]}, font=\small]
  % ----- generative models (left): TN + NN on top, quantum circuit below -----
  \node[draw=gray!50, rounded corners=3pt, fill=gray!4, minimum width=4.6cm,
        minimum height=3.9cm] at (1.25,0.0) {};
  \node[font=\bfseries\footnotesize, anchor=north] at (1.25,-2.15) {Generative model $q_\theta$};
  % tensor network (top left)
  \node[font=\footnotesize, anchor=south] at (0.05,1.55) {tensor network};
  \foreach \i in {0,...,3}{
    \node[circle, draw=green!45!black, fill=green!25, inner sep=2.1pt]
      (tn\i) at ({-0.62+\i*0.45},1.15) {};
    \draw ({-0.62+\i*0.45},1.01) -- ({-0.62+\i*0.45},0.75);
  }
  \draw (tn0) -- (tn1) -- (tn2) -- (tn3);
  % neural network (top right)
  \node[font=\footnotesize, anchor=south] at (2.5,1.55) {neural network};
  \foreach \i/\y in {0/0.75, 1/1.05, 2/1.35}{
    \node[circle, draw=blue!50, fill=blue!15, inner sep=1.8pt] (na\i) at (2.05,\y) {};
  }
  \foreach \i/\y in {0/0.9, 1/1.2}{
    \node[circle, draw=blue!50, fill=blue!15, inner sep=1.8pt] (nb\i) at (2.95,\y) {};
  }
  \foreach \a in {0,1,2}{ \foreach \b in {0,1}{ \draw[blue!40] (na\a) -- (nb\b); } }
  % quantum circuit (bottom)
  \node[font=\footnotesize, anchor=south] at (1.25,-0.5) {quantum circuit};
  \foreach \y in {-0.85,-1.25,-1.65}{
    \node[anchor=east, font=\footnotesize] at (0.0,\y) {$|0\rangle$};
    \draw (0.05,\y) -- (0.5,\y);
  }
  \node[draw, rounded corners=1.5pt, fill=green!12, minimum width=1.1cm,
        minimum height=1.35cm] at (1.05,-1.25) {$U(\theta)$};
  \foreach \y in {-0.85,-1.25,-1.65}{
    \draw (1.6,\y) -- (2.05,\y);
    \node[draw, rounded corners=1pt, minimum size=0.4cm, fill=white] at (2.28,\y) {};
    \draw[semithick] ($(2.28,\y)+(-0.11,-0.05)$) arc (200:-20:0.11);
    \draw[semithick,->] ($(2.28,\y)+(-0.01,-0.05)$) -- ++(0.085,0.1);
  }
  \draw[->, thick] (3.62,0.0) -- (4.95,0);
  \node[anchor=south, font=\footnotesize] at (4.28,0.1) {$x\sim q_\theta$};
  % ----- distribution space with low-loss sets (right) -----
  \begin{scope}[on background layer]
    \node[ellipse, draw=green!55!black, fill=green!20, fill opacity=0.55,
          minimum width=3.1cm, minimum height=3.3cm] (G) at (11.0,0) {};
    \node[ellipse, draw=blue!60, fill=blue!14, fill opacity=0.5,
          minimum width=3.3cm, minimum height=2.5cm] (L1) at (10.5,1.4) {};
    \node[ellipse, draw=orange!80!black, fill=orange!18, fill opacity=0.5,
          minimum width=3.3cm, minimum height=2.5cm] (L2) at (10.2,-1.5) {};
    \node[ellipse, draw=red!65!black, fill=red!12, fill opacity=0.5,
          minimum width=3.7cm, minimum height=2.8cm] (L3) at (8.3,0) {};
  \end{scope}
  \node[green!35!black, font=\footnotesize, align=center, anchor=west] at (12.7,0.15)
    {generalizes\\{\scriptsize low forward KL}};
  \node[blue!60!black, font=\footnotesize] at (9.85,3.0) {$\mathcal{L}_1$: likelihood};
  \node[orange!70!black, font=\footnotesize] at (9.55,-3.05) {$\mathcal{L}_2$: correlator};
  \node[red!65!black, font=\footnotesize] at (7.5,1.85) {$\mathcal{L}_3$: MMD$^2$};
  \node[star, star points=5, star point ratio=2.3, fill=green!55!black, inner sep=1.3pt] (ps) at (11.6,0.15) {};
  \node[anchor=west, font=\footnotesize] at (11.75,0.15) {$p^*$};
  \node[circle, fill=blue!70, inner sep=1.9pt] (m1) at (11.0,1.3) {};
  \node[circle, fill=orange!85!black, inner sep=1.9pt] (m2) at (10.7,-1.35) {};
  \node[font=\bfseries\large, red!70!black] (m3) at (7.5,0.05) {$\times$};
  \draw[->, densely dotted, thick, blue!65] (5.3,0.25) to[out=30,in=182] (m1);
  \draw[->, densely dotted, thick, orange!80!black] (5.3,-0.25) to[out=-30,in=178] (m2);
  \draw[->, densely dotted, thick, red!70!black] (5.3,0.0) to[out=0,in=205] (m3);
  \node[anchor=north, font=\footnotesize] at (4.45,-0.12) {minimize $\mathcal{L}_k$};
\end{tikzpicture}
\caption{\textbf{Training losses versus deployment generalization.} An optimizer fits a
generative model $q_\theta$ (a tensor network, a neural network, or a quantum circuit)
using quantities that can be evaluated without sampling the deployed model. The benchmark
tests whether small empirical values of the negative log-likelihood ($\mathcal{L}_1$), the fixed-order correlator
($\mathcal{L}_2$), and the full-kernel $\mathrm{MMD}^2$ ($\mathcal{L}_3$) objectives correspond to better
generalization (low forward KL; green).}
\label{fig:overview}
\end{figure*}
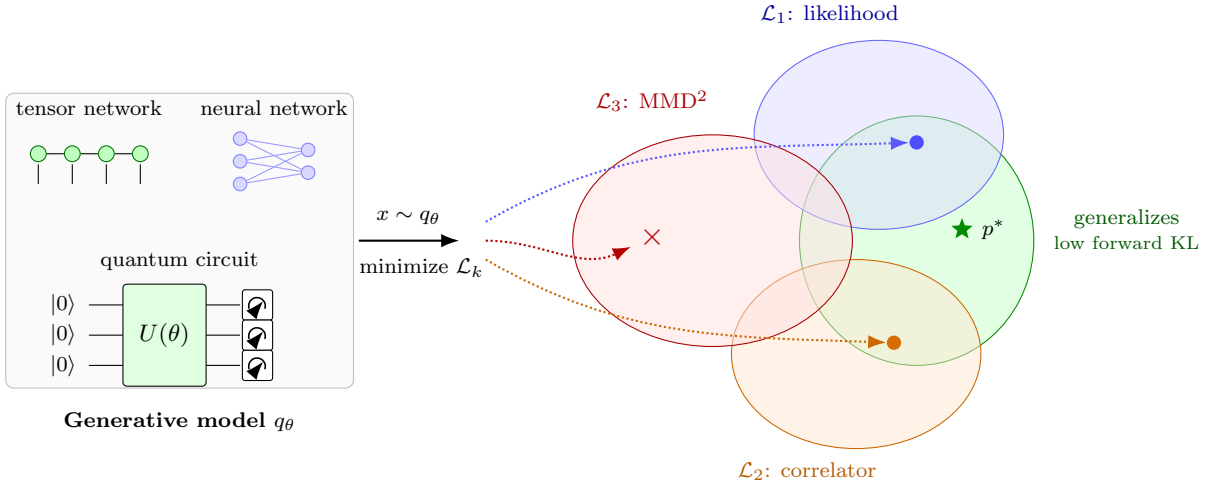

Generative models are used across science and industry, from image synthesis~\cite{goodfellow2014gan,ho2020ddpm}
and text generation~\cite{brown2020gpt3} to the design of new molecules~\cite{gomezbombarelli2018}.
Generative models are usually evaluated on whether the samples that they produce are novel, valid, and faithful to the target
distribution, rather than reproduce their training examples~\cite{carlini2023extracting,brown2019guacamol}. Quantum circuits provide a natural way to build such models. A quantum circuit Born machine
(QCBM) prepares a parametrized quantum state and generates samples by measuring it in the
computational basis~\cite{benedetti2019qcbm,liu2018mmdqcbm,coyle2020born}. Quantum correlations
can increase the expressive power of generative models~\cite{cheng2018information,gao2022enhancing}.
For certain circuit families, sampling from the output distribution is believed to be classically
intractable under standard complexity assumptions~\cite{bremner2011iqp,oszmaniec2022fermion}.
This possible sampling advantage has made QCBMs a candidate application for near-term quantum
computers~\cite{benedetti2019qcbm,coyle2020born}. QCBMs have been explored on tasks ranging from
finance to image synthesis, and parametrized quantum circuits are applied across machine
learning more broadly~\cite{cherrat2023hedging,mathur2025bayesian,raj2026scalable,raj2025quic}. For a review of quantum generative models beyond QCBMs and their applications, see~\cite{tian2022qglm}.

Training a QCBM directly on quantum hardware is however expensive, even with improved
training schemes~\cite{coyle2025density,coyle2026adaptive}. Each gradient estimate requires many
circuit evaluations, and barren plateaus can make typical gradients exponentially small in the
number of qubits~\cite{mcclean2018bp}.  One route around this cost pretrains the circuit classically as a matrix-product-state approximation and then fine-tunes it on hardware~\cite{rudolph2023synergistic}; the pretraining is classically efficient, though surpassing classical models still requires the added cost of hardware training. These costs have motivated a \emph{train classical, deploy
quantum} (TCDQ) strategy. The loss is evaluated and optimized on a classical computer, and the quantum
device is used only after training to draw samples~\cite{recio2025iqpopt,bako2025fbm}.  This separation is possible because evaluating the quantities used by the loss can be classically
tractable even when sampling from the full output distribution is believed to be classically hard.
Indeed, many models that provably avoid barren plateaus also allow their losses to be evaluated
efficiently on a classical computer~\cite{cerezo2025nogo}. Work on this strategy has mainly asked
whether the classical loss can be optimized~\cite{lerch2026iqpinit,rudolph2023synergistic} and
whether the deployed model can be reproduced by a classical surrogate~\cite{herrero2025born}.
Existing work has paid much less attention to ensuring that the trained circuit is a useful generator. Sampling may be
classically hard even if the model fails to produce unseen, valid data. We therefore ask a third
question: once the loss has converged, does the model generalize?

Gili et al.~\cite{gili2022} established a way to measure generalization in this setting. They distinguish an \emph{efficient learner}, which only reproduces the training distribution, from an \emph{efficient generator}, which produces new and valid samples. The two are separated by sampling the trained model, through four metrics: exploration, fidelity, rate, and coverage. Scoring a generative model by the novel, valid samples it produces is also long-standing practice in molecular design~\cite{brown2019guacamol}. In this work we take up the generalization question through both an empirical study and a theoretical analysis. Empirically, we benchmark thirteen classical and quantum generative models at up to $30$ qubits, on two application-inspired datasets, a cardinality-constrained family and a set of genomic single-nucleotide variants. We score each model by free sampling, and we compare its training loss against its coverage and its forward Kullback-Leibler divergence to the target distribution. Across the benchmark, sampling-based metrics order the models better than the moment-matching loss (Fig.~\ref{fig:overview}). Theoretically, we prove that a loss fixed by a prescribed set of low-order correlators cannot certify generalization on its own. Two distributions can drive such a loss to zero and still cover very different fractions of the valid set.

The paper is organized as follows. Section~\ref{sec:background} reviews quantum circuit Born machines and the TCDQ paradigm. Section~\ref{sec:methods} defines the deployed training losses as moment losses,
generalization through the forward KL, and its sample-based metrics, coverage and fidelity.
Section~\ref{sec:theory} proves that a moment loss of bounded capacity admits exact global
minimizers of exponentially small coverage, and locates the failure of the full MMD in its
empirical minimizer, the memorizer. Section~\ref{sec:results} reports the benchmark, ranking the models by loss, coverage, and forward KL on the cardinality-constrained family (to $30$ qubits by state-vector simulation) and the genomic variants, and testing robustness across circuit gate count, kernel bandwidth, and a classical baseline. 

\section{Background and related work}
\label{sec:background}

\paragraph{Generative modeling.} A generative model learns a distribution from a finite set of
samples and produces new data from it. Generative adversarial networks train a generator against
a discriminator~\cite{goodfellow2014gan}; autoregressive and diffusion models maximize the
likelihood of the data; the trained models generate text~\cite{brown2020gpt3},
images~\cite{ho2020ddpm}, and molecules~\cite{gomezbombarelli2018}. How well such a model
generalizes is judged on data beyond its training set. Arora et al.\ showed that the adversarial
objective can be driven to its optimum by a generator supported on roughly as many points as the
discriminator has parameters, so a small training objective does not imply that the learned
distribution is close to the target~\cite{arora2017}, and support-size measurements on trained
GANs confirmed the effect~\cite{arorazhang2018}. Evaluation therefore moved to the samples.
Precision and recall separate the quality of generated data from the fraction of the target the
model covers~\cite{sajjadi2018,naeem2020}, and memorization tests check that generated data are
new~\cite{carlini2023extracting}. Our study asks the corresponding questions for quantum
generative models, in the spirit of careful benchmarking of quantum learning
models~\cite{bowles2024benchmarking,schuld2022advantage}, building primarily from the work
of Gili et al.~\cite{gili2022}, and formalizing our generalization derivations and results
in the context of Arora et al~\cite{arora2017}.

\paragraph{Quantum circuit Born machines.} A Born machine encodes a probability
distribution as the measurement distribution of a quantum state, $p_\theta(x) =
|\langle x | \psi_\theta\rangle|^2$. Trained variationally, they are a standard model for
quantum generative learning~\cite{benedetti2019qcbm,liu2018mmdqcbm,coyle2020born}, and have been
trained on trapped-ion hardware on tasks up to high-resolution
imagery~\cite{zhu2019training,rudolph2022digits,leytonortega2021robust}. For such models, sampling
is native to quantum hardware, and for circuit families such as
instantaneous quantum polynomial-time (IQP) circuits the output distribution is classically
hard to sample under standard complexity assumptions~\cite{bremner2011iqp}. Whether such
distributions are efficiently learnable at all is an important research direction: Clifford outputs
are learnable, yet a single non-Clifford gate can make distribution learning
hard~\cite{sweke2021learnability,hinsche2021learnability,hinsche2023tgate}.

\paragraph{Train classical, deploy quantum.} Training Born machines on hardware is challenging,
owing to barren plateaus~\cite{mcclean2018bp} and the cost of estimating gradients
from samples. A way around this is to keep training entirely classical and use the quantum
device only for inference. Recio-Armengol et al.\ train IQP Born machines on the moment-matching
loss (MMD$^2$) at up to a thousand qubits by writing the loss as a classical mixture of Pauli-$Z$ word
expectations~\cite{recio2025iqpopt,rudolph2024trainability}, each computable in polynomial
time via den Nest's estimator~\cite{dennest2010iqp}. Bak\'o et al.\ propose fermionic Born
machines (FBMs): magic (non-Gaussian) input states evolved by matchgate (free-fermion) circuits,
trained by matching constant-locality Pauli-$Z$ string correlators~\cite{bako2025fbm}, building on
the Fermion Sampling advantage scheme~\cite{oszmaniec2022fermion}. Rudolph et al.\ pretrain a parametrized circuit from a
classically optimized tensor network~\cite{rudolph2023synergistic}; related classical-training
protocols appear in Refs.~\cite{kasture2023protocols,kyriienko2024protocols}. The training loss uses only
classically cheap observables, while sampling from the trained model can remain
classically hard for suitable circuit families.

\paragraph{Trainability and simulability.} Most analysis of this paradigm studies whether the
moment-matching loss is optimizable or whether the model is classically surrogatable. Lerch et al.\
characterize when the MMD$^2$ loss of an IQP Born machine has barren plateaus and when
data-dependent initialization restores trainability~\cite{lerch2026iqpinit}. Herrero-Gonz\'alez
et al.\ identify the Born-rule probability as a Fourier series of Pauli-$Z$ correlators, give
conditions under which a correlator-described model can be classically surrogated, quantify the
discrepancy between classically trained and quantumly deployed parameters, and note that an MMD$^2$
loss whose kernel omits the relevant correlators resolves only the selected correlator
subset~\cite{herrero2025born}; the expressivity of such quantum Fourier models is itself
constrained~\cite{mhiri2025expressivity}. T\"uys\"uz et al.\ show that a fixed MMD$^2$ kernel weights the
Walsh--Hadamard spectrum toward low Hamming weight, so an MMD$^2$-trained model matches only
low-order correlations and departs from correlated, multimodal targets past a low-weight cutoff,
and a multi-kernel MMD$^2$ does not remove this; they instead train a Fourier-encoded circuit by a
classically tractable marginal likelihood that scales to a thousand qubits~\cite{tuysuz2026dqgm}.

\paragraph{Evaluating generalization.} Scores computed against the \emph{empirical training}
distribution (its KL, total variation, or moment-matching loss to the training data) measure
agreement with the training set and cannot by themselves separate memorization from
generalization. This is distinct from the forward KL divergence $\mathrm{KL}(p^*\,\|\,
q_\theta)$ to the \emph{known target}, which we use later as the gold standard for generalization
precisely because the target distribution, $p^*$, is known by design. Gili et al.\ introduce a
sample-based framework (EFRC) on cardinality-constrained datasets, where validity is cheap to
verify, and define the metrics we use, including the coverage of the unseen valid
set~\cite{gili2022,hibat2024framework,gili2023generalize}. We adopt coverage and the forward KL as our
generalization measures throughout. Two of the classical models we evaluate are the tensor-network
Born machine (TNBM) of Han et al.~\cite{han2018mps}, which trains by likelihood, and the
generative moment-matching network (GMMN)~\cite{li2015gmmn,dziugaite2015}, trained on the same
moment-matching loss (MMD$^2$); the GMMN lets us separate the effect of the loss from that of the
quantum platform. We measure coverage on
both the cardinality-constrained datasets and a real-world dataset of genomic single-nucleotide
variants, where the valid set is the set of observed sequences. For the genomic data we also
report standard population-genetics diagnostics, per-locus allele-frequency correlation and
pairwise linkage disequilibrium, alongside coverage~\cite{walonoski2018synthea}.

\section{Framework}
\label{sec:framework}
\label{sec:methods}

\paragraph{Generative modeling.} A generative task is specified by a distribution $p^*$ on
$\{0,1\}^N$. Its support $S=\supp(p^*)$ is the set of \emph{valid} strings. The learner
observes a finite training set $T=(X_1,\dots,X_n)$ drawn from $p^*$ and its empirical
distribution
\begin{equation}
\label{eq:emp}
\tilde p_T=\frac{1}{n}\sum_{i=1}^{n}\delta_{X_i}.
\end{equation}
It returns a model distribution $q_\theta$. In a train-classical, deploy-quantum (TCDQ)
workflow, quantities used during optimization are evaluated classically, whereas the trained
model is sampled from the quantum circuit at deployment. The samples it produces beyond $T$
determine its usefulness.

\paragraph{Quantum circuit models.} A quantum circuit Born machine (QCBM) prepares
$U(\theta)|0\rangle$ and generates
\begin{equation}
\label{eq:born}
q_\theta(x)=|\langle x|U(\theta)|0\rangle|^2 .
\end{equation}
Sampling is native to the hardware, but evaluating $q_\theta(x)$ requires estimating an
exponentially small projector overlap, so the model's probabilities are not efficiently
accessible~\cite{rudolph2024trainability}. For the IQP and fermionic circuit families
considered here, selected Pauli-$Z$ correlators can be evaluated classically even when
sampling from the full output distribution is believed to be classically
difficult~\cite{recio2025iqpopt,bako2025fbm}.

\paragraph{Losses.} When we train a quantum generative model, we optimize a loss
$\mathcal{L}$ that measures disagreement between the model and a reference distribution.
Three versions of any such loss must be distinguished:
\begin{equation}
\label{eq:objects}
\begin{aligned}
\mathcal{L}_*(q)&:=\mathcal{L}(p^*,q), &&\text{population discrepancy,}\\
\mathcal{L}_T(q)&:=\mathcal{L}(\tilde p_T,q), &&\text{empirical training loss,}\\
\widehat{\mathcal{L}}_{T,B}(q)&\phantom{:=\mathcal{L}(p,q),} &&\text{finite-budget estimator of }\mathcal{L}_T(q).
\end{aligned}
\end{equation}
Here $B$ counts the queries the estimator spends per evaluation, correlator draws and
model samples together. A small $\widehat{\mathcal{L}}_{T,B}$ can mean that the estimator is
accurate and that the empirical loss has been optimized. It does not by itself imply that
$q_\theta$ is close to $p^*$. For the circuit families the paradigm uses, the Pauli-$Z$ correlators are
classically computable~\cite{dennest2010iqp,recio2025iqpopt}, while log-probabilities are not
available to the precision a likelihood needs~\cite{nivandennest2013}.

The deployed losses share one structural property, which Section~\ref{sec:theory} uses in its
general form:

\begin{definition}[Moment loss of capacity $d$]
\label{def:momentloss}
A loss is a \emph{moment loss of capacity $d$} if it depends on the model only through the
expectations $\langle O_1\rangle_q,\dots,\langle O_d\rangle_q$ of bounded observables
$O_j:\{0,1\}^N\to[-1,1]$, and attains its minimum exactly when all $d$ match the reference.
\end{definition}

\noindent The widely used TCDQ losses are all moment losses, and their observables are
the Pauli-$Z$ strings $Z_A$. For $A\subseteq[N]$ write
$\chi_A(x)=(-1)^{\sum_{i\in A}x_i}$ for the parity of $A$, and
$\ZA_q=\sum_x q(x)\chi_A(x)$ for the correlator of the model. Every deployed loss is a
weighted square difference of correlators,
\begin{equation}
\label{eq:genloss}
\mathcal{L}_{w,\mathcal{F}}(p,q)\;=\;\sum_{A\in\mathcal{F}} w(A)\,
\big(\ZA_{p}-\ZA_{q}\big)^2,
\end{equation}
with the reference $p$ set to $p^*$, $\tilde p_T$ or a batch, per Eq.~\eqref{eq:objects}. The
index set $\mathcal{F}$ specifies which correlators the loss compares, the weight $w$
sets the penalty on each mismatch, and the capacity of Definition~\ref{def:momentloss} is
$d=\dim\mathrm{span}\,\mathcal{F}$. The three deployed losses are three choices of
$(w,\mathcal{F})$, with $D_L:=\sum_{\ell\le L}\tbinom{N}{\ell}=O(N^L)$ counting the
correlators of order at most $L$:
\begin{equation}
\label{eq:threelosses}
\begin{aligned}
\text{FBM correlator~\cite{bako2025fbm}:}&\quad
\mathcal{L}_{1,\{|A|\le L\}}, && d=D_L,\\
\text{truncated MMD:}&\quad \mathcal{L}_{w_\sigma,\{|A|\le L\}}, && d=D_L,\\
\text{full MMD~\cite{recio2025iqpopt}:}&\quad \mathcal{L}_{w_\sigma,2^{[N]}}, && d=2^N.
\end{aligned}
\end{equation}
The first is evaluated exactly through Pfaffians~\cite{bako2025fbm}. The third is never
evaluated exactly: the sampled-spectrum estimator of Ref.~\cite{recio2025iqpopt} estimates it
without bias by redrawing $A\sim w_\sigma$ at every update, an instance of the estimator
$\widehat{\mathcal{L}}_{T,B}$ of Eq.~\eqref{eq:objects}. These losses are
\emph{training surrogates}: they stand in for the inaccessible likelihood during
optimization, and they are the objectives under which the paradigm
trains~\cite{liu2018mmdqcbm,recio2025iqpopt,bako2025fbm}.

The weight $w_\sigma$ comes from a kernel. The MMD compares two distributions through their
mean embeddings $\mu_q=\mathbb{E}_{q}[\phi(x)]$, and the loss is $\|\mu_{p}-\mu_q\|^2$,
written with the kernel alone as
\begin{equation}
\label{eq:mmddef}
\begin{aligned}
\MMD^2_k(p,q)={}&\mathbb{E}_{x,x'\sim p}\,k(x,x')+\mathbb{E}_{y,y'\sim q}\,k(y,y')\\
&-2\,\mathbb{E}_{x\sim p,\,y\sim q}\,k(x,y).
\end{aligned}
\end{equation}
On $N$-bit strings the kernel is the Gaussian--Hamming kernel
$k_\sigma(x,y)=e^{-d_H(x,y)/2\sigma^2}$, with $d_H$ the Hamming distance and $\sigma$ the
bandwidth.\footnote{We follow the convention of Ref.~\cite{recio2025iqpopt}, in which $\sigma$
is a standard deviation.} In the Fourier basis of the cube~\cite{odonnell2014},
\begin{equation}
\label{eq:mmdcorr}
\begin{aligned}
\MMD^2_\sigma(p,q)&=\mathcal{L}_{w_\sigma,\,2^{[N]}}(p,q),\\
w_\sigma(A)&=p_\sigma^{\,|A|}(1-p_\sigma)^{N-|A|},\\
p_\sigma&=\tfrac12\big(1-e^{-1/2\sigma^2}\big),
\end{aligned}
\end{equation}
so the full MMD compares every correlator, with weights that peak at order $p_\sigma N$, which
is $0.197N$ at $\sigma=1$, and fall exponentially above
it~\cite{recio2025iqpopt,tuysuz2026dqgm,rudolph2024trainability}. Since
$w_\sigma(A)>0$ for every $A$, the kernel is characteristic: the population MMD vanishes
only at $q=p^*$~\cite{fukumizu2008,sriperumbudur2010}.

\begin{figure*}[t]
\centering
\includegraphics[width=0.7\textwidth]{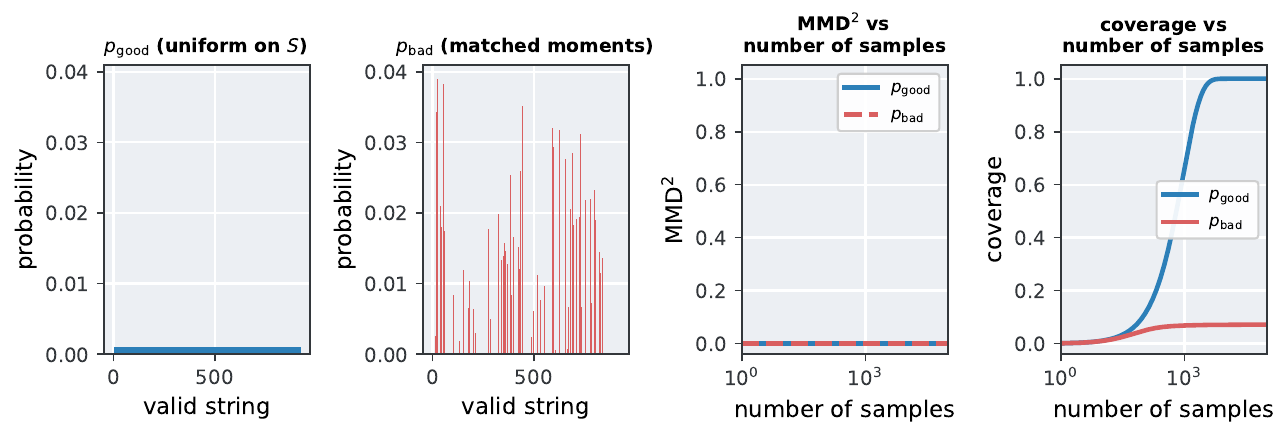}
\caption{\textbf{A truncated moment loss admits exact minimizers with different support coverage.}
The distributions $p_{\rm good}$ and $p_{\rm bad}$ match every correlator through
order $L$ (Corollary~\ref{cor:covbound}), yet their support coverages are $1$ and $0.07$.
Statements in Section~\ref{sec:theory}, proofs in Appendix~\ref{app:theory}.}
\label{fig:demo}
\end{figure*}

\begin{figure}[t]
\centering
\includegraphics[width=\columnwidth]{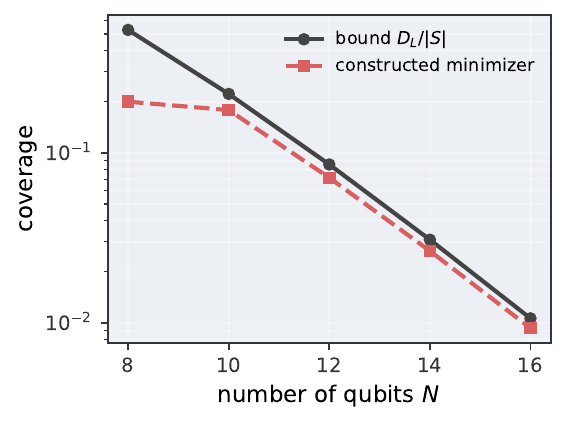}
\caption{\textbf{The low-coverage minimizer of Corollary~\ref{cor:covbound} is explicit.} For $N=8$ to
$16$ at $L=2$, a linear program constructs a distribution matching every order-$\le L$ correlator
of the uniform target. Realized coverage of the constructed minimizer (red) and the bound
$D_L/|S|$ of Corollary~\ref{cor:covbound} (gray) versus $N$, on a logarithmic scale; the truncated loss
is numerically zero at every point.}
\label{fig:lp}
\end{figure}

Each deployed loss is the square of an integral probability
metric~\cite{muller1997,gretton2012mmd,sriperumbudur2009} with discriminator class
$\mathcal{F}$ in the sense of Arora et al.~\cite{arora2017} and capacity $d$.

\paragraph{Generalization and its sample-based evaluation.} Generalization is the population
goal of a generative model: closeness to the target on the whole valid set, including the
strings the training data never showed it~\cite{arora2017,gili2022}. We state it with the
forward KL divergence.

\begin{definition}[Generalization]
\label{def:gen}
A model $q_\theta$ \emph{generalizes at level $\varepsilon$} if
\begin{equation}
\label{eq:fkl}
\mathrm{KL}(p^*\,\|\,q_\theta)=\sum_{x\in S} p^*(x)\,\log\frac{p^*(x)}{q_\theta(x)}
\;\le\;\varepsilon .
\end{equation}
\end{definition}

\noindent The forward KL is finite only if the model assigns probability to every valid
string, and small only if it weights them as the target does. That the forward KL blows up under support mismatch is the property
Arjovsky and Bottou use to motivate weaker distances for GAN
training~\cite{arjovsky2017towards}; here it is what makes coverage the natural
sample-based metric. In deployment the divergence cannot be evaluated: $q_\theta(x)$ is not efficiently
accessible, and $p^*$ is unknown outside designed benchmarks; Section~\ref{sec:results} computes the
divergence offline for its benchmark targets. What is available at deployment is a quantity
computed from the samples the model produces, under a stated budget.

\begin{definition}[Sample-based generalization metric]
\label{def:metric}
Let $G(q_\theta;T)$ be a quantity determined by the trained model and the training
set. A
statistic $\widehat G_Q$, computed from the first $Q$ terms of a sequence
$Y_1,Y_2,\dots$ of independent samples of $q_\theta$, is a \emph{sample-based
generalization metric} for $G$ if $\widehat G_Q\to G(q_\theta;T)$ as $Q\to\infty$, with
probability one.
\end{definition}

\noindent Write $U=S\setminus T$ for the unseen valid strings, assumed nonempty. We use two such
metrics: the coverage, for diversity over $U$, and the fidelity $F$, the fraction of
generated samples that are valid; the fidelity converges to the valid mass
$q_\theta(S)$. The \emph{support coverage}
\begin{equation}
\label{eq:cinf}
C_\infty(q_\theta)\;=\;\frac{|\supp(q_\theta)\cap U|}{|U|}
\end{equation}
is the fraction of the unseen valid strings with positive probability under the model.
The \emph{measured coverage} at budget $Q$ is the statistic
$\widehat C_Q(q_\theta)=|\{Y_1,\dots,Y_Q\}\cap U|/|U|$, with expectation
\begin{equation}
\label{eq:cbar}
\mathbb{E}\,\widehat C_Q(q_\theta)\;=\;\frac{1}{|U|}\sum_{x\in U}
\big[\,1-(1-q_\theta(x))^{Q}\,\big].
\end{equation}
Every sampled string lies in the support, so $\widehat C_Q\le C_\infty$ at every $Q$,
and $\widehat C_Q\to C_\infty$ with probability one. The shortfall
$C_\infty-\mathbb{E}\,\widehat C_Q=|U|^{-1}\sum_{x\in\supp(q_\theta)\cap U}
(1-q_\theta(x))^{Q}$ is small once $Q$ is large compared with the reciprocal of the
smallest positive model probability on $U$. Whether measured coverage ranks
models as the forward KL does is an empirical question, and Section~\ref{sec:results} reports the observed
rank correlations. The classical and the quantum literature track the generalization of
generative models with coverage in one form or
another~\cite{naeem2020,sajjadi2018,gili2022}. Coverage alone suffices
when the architecture confines the support to $S$, and we report coverage and fidelity
together otherwise.

\section{Moment losses do not certify generalization}
\label{sec:theory}

In this section we show that exact minimizers of the deployed losses can have
exponentially small coverage. Our results are statements about the population loss
$\mathcal{L}_*$ of Eq.~\eqref{eq:objects}, the best case for the loss. The two finite
objects are weaker: for a characteristic kernel the exact minimizer of the empirical loss
$\mathcal{L}_T$ is the memorizer $\tilde p_T$, and a small estimate
$\widehat{\mathcal{L}}_{T,B}$ implies only that the estimator is accurate. The population
values are estimable from polynomially many samples~\cite{arora2017}. We show results on
the cardinality task of Section~V: $S=\{x:|x|=N/2\}$ with
$|S|=\binom{N}{N/2}=\Theta(2^N/\sqrt N)$, $p^*$ uniform on $S$, and $T=\varnothing$.

\begin{theorem}[Moment losses of bounded capacity admit exact sparse minimizers]
\label{thm:main}
Let $p^*$ have full support on a finite valid set $S$, and let $\mathcal{L}$ be a
moment loss of capacity $d$ with $d+1<|S|$. Then there is an exact global minimizer
$q_{\mathrm{sparse}}$ of $\mathcal{L}_*$, supported inside $S$ on at most $d+1$ strings,
with
\begin{equation}
\label{eq:sparse}
\begin{gathered}
\mathcal{L}_*(q_{\mathrm{sparse}})=\mathcal{L}_*(p^*),\qquad
\mathrm{KL}(p^*\,\|\,q_{\mathrm{sparse}})=\infty,\\
\widehat C_Q(q_{\mathrm{sparse}})\;\le\;C_\infty(q_{\mathrm{sparse}})\;\le\;
\frac{d+1}{|S|}\ \ \text{for every }Q.
\end{gathered}
\end{equation}
\end{theorem}

\noindent The target is also an exact minimizer, and the loss takes the same value on both,
so the loss value cannot select between them. By Eq.~\eqref{eq:sparse},
$q_{\mathrm{sparse}}$ does not generalize at any level $\varepsilon$
(Definition~\ref{def:gen}). The value of a training surrogate
therefore carries no guarantee of coverage at deployment, and evaluation must come from
samples. Model structure or optimizer bias may
still select a good minimizer among the exact minimizers. The statement constrains distributions; whether a given
circuit family realizes the sparse minimizer, and whether training reaches one, are separate
questions, and Section~\ref{sec:results} addresses them empirically.

For the deployed losses on the cardinality task the bound is explicit.

\begin{corollary}[Fixed-order correlators on the cardinality task]
\label{cor:covbound}
The order-$\le L$ losses have $d=D_L=\sum_{\ell\le L}\binom{N}{\ell}=O(N^L)$, and
the constant function lies in the span of their observables, so on the uniform cardinality
target they admit an exact global minimizer $q_L$ supported on at most $D_L$ strings, with
\begin{equation}
\label{eq:corgamma}
\begin{gathered}
\widehat C_Q(q_L)\;\le\;C_\infty(q_L)\;\le\;\frac{D_L}{\binom{N}{N/2}}
\;=\;O\!\Big(\frac{N^{L+1/2}}{2^N}\Big)\\
\text{for every }Q.
\end{gathered}
\end{equation}
\end{corollary}

The proofs are in Appendix~\ref{app:proof}; the support count is $d+1$ in general and
improves to $d$ when the constant function lies in the span of the observables, as it does
for the deployed losses. The theorem is the
Boolean-cube form of the GAN analysis of Arora et al.~\cite{arora2017}, in which a
discriminator class of bounded capacity is driven to its optimum by generators supported on
roughly capacity-many points, an effect later confirmed by support measurements on trained
GANs~\cite{arorazhang2018}. Our result does not exclude a TCDQ loss that depends on more than a bounded list of
expectations. In particular, Tuysuz et al.\ train a
Fourier-structured model on a classically estimated likelihood, which is a moment loss of no
finite capacity, and pay with a restricted circuit family whose visible marginals
factorize~\cite{tuysuz2026dqgm}.

We construct the minimizer $q_L$ of Corollary~\ref{cor:covbound} explicitly with a linear program
(Fig.~\ref{fig:lp}). At $N=12$ it finds a distribution that matches every correlator of
$p^*$ through order $L=2$ on a support of $66$ strings. Its truncated loss is $\sim\!10^{-28}$, its coverage is
$0.07$ against the target's $1.00$, and its forward KL is infinite
(Table~\ref{tab:worked}). Figure~\ref{fig:demo} visualizes the two distributions.

Theorem~\ref{thm:main} does not apply to the full MMD: the kernel is characteristic
(Section~\ref{sec:framework}), so $p^*$ is the unique exact minimizer of the population
loss $\mathcal{L}_*$~\cite{liu2018mmdqcbm}, and the full MMD remains a principled training
surrogate. Training, however, minimizes the empirical and stochastic objects of
Eq.~\eqref{eq:objects}. The memorizer $\tilde p_T$ attains $\mathcal{L}_T(\tilde p_T)=0$
with $C_\infty(\tilde p_T)=0$. An empirical MMD value therefore cannot substitute for
sample-based evaluation, and Section~\ref{sec:results} evaluates every trained model from its samples.

\section{Benchmarking generalization}
\label{sec:results}

We train thirteen models on two datasets and report how
well each generalizes. The datasets probe opposite ends of the correlation spectrum: the genomic target is
dominated by low-order statistics (allele frequencies and pairwise linkage), which a moment loss
can represent, while the cardinality constraint lives in correlators of every order
(Section~\ref{sec:theory}). For every model we compare its training loss, the eval MMD$^2$, against two
direct measures of generalization, coverage and the forward KL, and ask whether a low loss tracks
them. Because the target $p^*$ is known by design, for the small systems ($N{=}16$ and $N{=}20$)
we rank the models by each score and compare the
rankings; the cardinality task is pushed to $N=30$ by full state-vector simulation, confirming
Section~\ref{sec:theory} in practice: models with nearly the same MMD$^2$ loss differ widely in coverage.
Table~\ref{tab:benchmark} reports MMD$^2$ and coverage on both datasets. Throughout
this section, $C$ denotes the measured coverage $\widehat C_Q$ of Section~\ref{sec:framework}
at the stated budget $Q$.

\begin{table*}[t]
\centering
\caption{\textbf{A benchmark of quantum and classical generative models, each scored by both its
training loss and its ability to generate new, valid data.} We present results on two datasets together: the
cardinality task at $N=30$ ($\varepsilon\approx3\times10^{-4}$, corresponding to $|T|=50$K training strings; $Q=10^5$;
three seeds) by full state-vector simulation, and genomic single-nucleotide variants at $N=20$
($\varepsilon=0.20$, corresponding to $|T|=796$; $Q=3{,}980$; five seeds). Values are medians over seeds. Rows group models into likelihood-trained and
moment-trained families. We report eval MMD$^2$ (Gaussian--Hamming kernel, $\sigma{=}1$) for both;
fidelity $F$ and normalized coverage $\tilde C=C/C^*$ for the cardinality columns (the raw coverage
is query-budget-limited at $N=30$); and coverage $C$ for the genomic. Colors encode within-column
rank, from green for better values to red for worse values. Per-size cardinality and
per-$\varepsilon$ results, including the median-bandwidth IQP variant, are in
Appendix~\ref{app:extra}. The active and magic FBMs emit only even-weight strings, so their cardinality entries
(double dagger) are at $N=28$, HW$=14$ ($\varepsilon\approx1.2\times10^{-3}$, $Q=10^5$, three
seeds), the largest even-weight size at which the magic FBM is sampled exactly
(Appendix~\ref{app:magic}).}
\label{tab:benchmark}
\resizebox{0.7\textwidth}{!}{\input{tables/master_combined.tex}}
\end{table*}

\begin{figure*}[t]
\centering
\includegraphics[width=\textwidth]{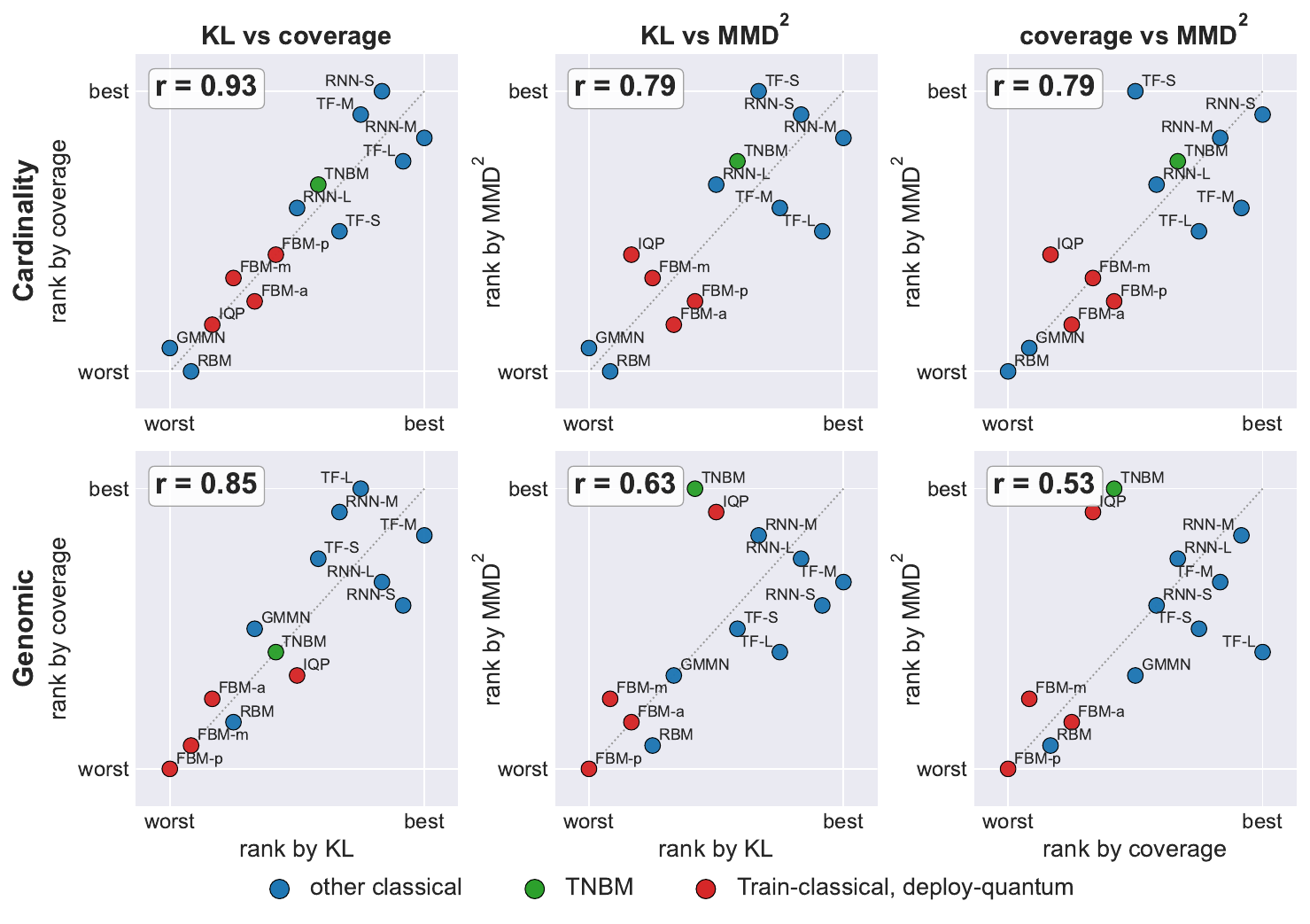}
\caption{\textbf{Rank correlations among MMD$^2$, coverage, and forward KL.}
Each panel ranks the same thirteen models by two scores, best in the top-right corner, for the cardinality
data (top) and the genomic data (bottom), both at $N=16$; each score is a median over seeds, and
$r$ is Spearman's rank correlation between the two scores. The left column pairs the two
generalization measures (coverage and the forward KL); the middle and right columns pair each
against the moment-matching loss MMD$^2$. In the labels, TF is the transformer and S/M/L are the
small, medium, and large variants. All rankings here are computed at
$N=16$, where exact model probabilities make the forward KL available; Table~\ref{tab:benchmark}
reports $N{=}30$ and $N{=}20$.}
\label{fig:rank}
\end{figure*}

\begin{figure*}[t]
\centering
\includegraphics[width=0.86\textwidth]{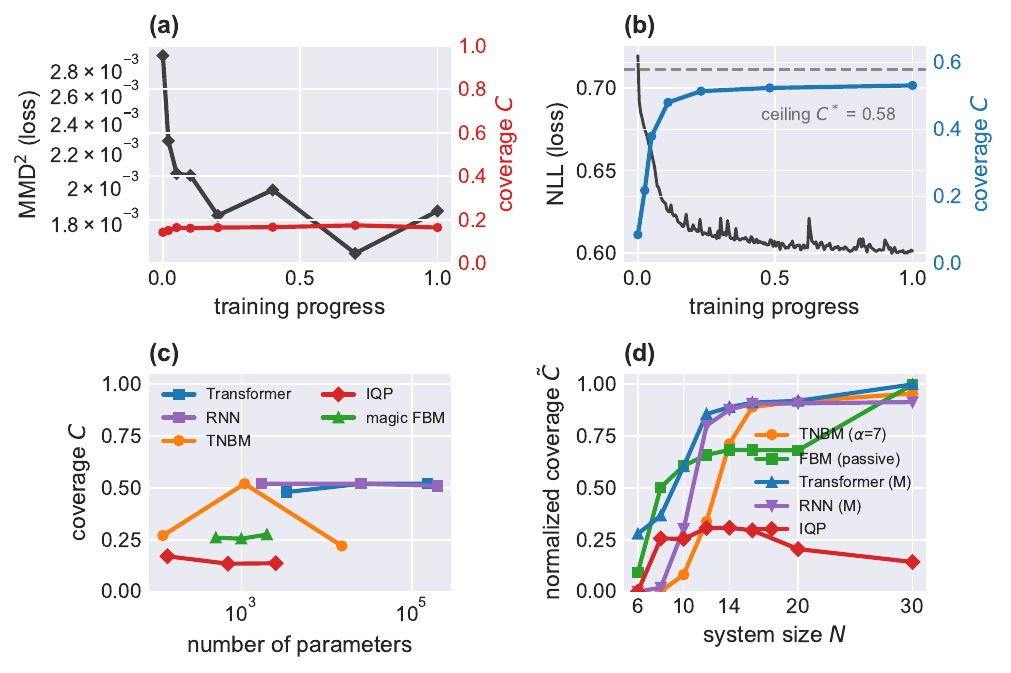}
\caption{\textbf{Loss and coverage decouple during training and across model size.}
\emph{(a)} IQP at $N=16$, $\varepsilon{=}0.10$: the loss (MMD$^2$, gray) falls while coverage
(red) stays flat near the random-bit coverage floor at which the initialization starts it.
\emph{(b)} a likelihood-trained
transformer: the loss (NLL, gray) falls and coverage (blue) rises to the finite-sample maximum
$C^*$. \emph{(c)} coverage versus parameter count at $N=16$: the classical transformer and RNN and
the TNBM rise with size, while IQP and the magic FBM do not. \emph{(d)}
normalized coverage $\tilde C=C/C^*$ versus system size $N$ up to $30$ (raw $C$ is not comparable
across $N$ because $|S|$ grows exponentially and the query budget changes with $N$): the
likelihood-trained and structural models increase toward $\tilde C=1$ while IQP remains below
$0.31$ throughout and falls to $0.14$ at $N=30$. Throughout, loss is gray and coverage takes
each model's group colour.}
\label{fig:mechanism}
\end{figure*}

\begin{table}[t]
\centering
\caption{Effect of model capacity at $N=16$, $\varepsilon=0.10$ (three seeds; two for the
TNBM at $\chi=7$). Each family has
its own size axis (in parentheses); parameters and coverage $C$ run from the smallest to the
largest setting. Enlarging IQP does not raise its coverage, while the TNBM rises then
over-fits, as illustrated in Fig.~\ref{fig:mechanism}c. The free-fermion FBMs have no capacity
parameter.}
\label{tab:capacity}
\setlength{\tabcolsep}{4pt}
\resizebox{\columnwidth}{!}{%
\begin{tabular}{lcc}
\toprule
family (size axis) & params (small$\to$large) & $C$\\
\midrule
Transformer (w/depth)   & $3.4$k/$25.6$k/$150$k & $0.48/0.52/0.52$\\
RNN (hidden dim)        & $1.7$k/$25.3$k/$199$k & $0.52/0.52/0.51$\\
TNBM (bond $\chi$)       & $120/1088/15016$      & $0.27/0.52/0.22$\\
IQP (gate weight)       & $136/696/2516$        & $0.17/0.13/0.14$\\
FBM (passive)           & $120$                 & $0.40$\\
FBM (active)            & $496$                 & $0.24$\\
magic FBM (FLO layers)  & $500/996/1988$        & $0.26/0.26/0.28$\\
\bottomrule
\end{tabular}}
\end{table}

\subsection{Datasets, models, and protocol}
\label{sec:setup}

\begin{definition}[Cardinality-constrained dataset]
\label{def:dataset}
For even $N$, the \emph{valid set} is the set of bitstrings of Hamming weight $N/2$,
\begin{equation}
S=\{x\in\{0,1\}^N : |x|=N/2\},\quad |S|=\binom{N}{N/2},
\end{equation}
and the \emph{target} $p^*$ is the uniform distribution on $S$. We use $N\in\{16,20,30\}$, for
which $|S|=12{,}870$, $184{,}756$, and $1.55\times10^8$.
\end{definition}

\noindent The fixed-weight rule couples all $N$ bits, so a model must capture their joint
structure. The valid set is known in full, so coverage can be measured exactly.

\begin{definition}[Genomic variant dataset]
\label{def:genomic}
From $M=5{,}008$ observed single-nucleotide-variant sequences~\cite{yelmen2021artificial}, taken
from the PennyLane quantum dataset library~\cite{bergholm2018pennylane}, we keep the first $N$ loci
of each sequence and deduplicate; the distinct observed sequences $x\in\{0,1\}^N$ form the
\emph{valid set}
\begin{equation}
S=\{x : x\in\mathcal{D}\},
\end{equation}
with $|S|=2{,}716$ at $N{=}16$ and $3{,}980$ at $N{=}20$. We define the \emph{target} $p^*$ as the
uniform distribution on $S$, the empirical distribution of the distinct sequences, and a sequence is \emph{valid} if it lies in $S$.
\end{definition}

\noindent The genomic valid set has no fixed Hamming weight: for example, at $N{=}20$ the modal
weight is $12$ and the data span weights from $3$ to $19$. The target $p^*$ is non-uniform, with per-locus allele
frequencies and linkage between loci.

\paragraph{Models.} We evaluate thirteen models. The likelihood-trained group consists of the
TNBM~\cite{han2018mps}, three transformers, three GRU RNNs, and an RBM. The MMD-trained group
consists of the IQP Born machine~\cite{recio2025iqpopt}, the magic FBM~\cite{bako2025fbm}, and
the classical GMMN~\cite{li2015gmmn,dziugaite2015}. The active ($SO(2N)$) and passive FBMs use
fixed-locality correlator losses; the passive FBM also restricts its support to the
fixed-Hamming-weight sector, so it is confined to the cardinality valid set by construction
($F=1$). The active FBM and the magic FBM prepare even-parity fermionic states and evolve
them by parity-preserving matchgates, so their output distributions are supported on even
Hamming weights (Appendix~\ref{app:hyper}); their cardinality entries are therefore taken at
$N=28$, where the target weight is even (Appendix~\ref{app:magic}). The IQP Born machine uses the marginal-matching
initialization at $N=16$ and $N=20$ and the magnetization initialization at $N=30$
(Appendix~\ref{app:init}). Appendix~\ref{app:hyper} gives the architectures and optimization settings.

\paragraph{Training protocol.} From a dataset we draw a training set $T$ at
fraction $\varepsilon$: for the cardinality task, $|T|=\lfloor\varepsilon|S|\rfloor$ valid strings
drawn uniformly without replacement from $S$ ($\varepsilon\in\{0.01,0.05,0.10\}$); for the genomic
task, $|T|=\lfloor\varepsilon|S|\rfloor$ distinct sequences drawn uniformly without replacement
from $S$
($\varepsilon\in\{0.01,0.05,0.10,0.20\}$). A seed fixes the training subset: at each
$(\varepsilon,\text{seed})$ every model trains on the same $T$, and $T$ is resampled across
seeds, so the spread reported across seeds reflects the choice of training subset. We use three
seeds for the cardinality task and five for the genomic task. Training fixes the
parameters by minimizing a loss $\mathcal{L}$ that compares $q_\theta$ to $T$
(Section~\ref{sec:methods}),
\begin{equation}
\theta^\star=\arg\min_\theta \mathcal{L}(q_\theta,T),
\end{equation}
with an iterative optimizer: gradient descent for the autoregressive, IQP, FBM, and GMMN models,
two-site DMRG for the TNBM, and contrastive divergence for the RBM. Autoregressive models run
$200$ epochs with early stopping; the IQP model runs $500$ gradient steps and the FBM
models $1{,}500$ to $2{,}000$. Full hyperparameters are in Appendix~\ref{app:hyper}.

\paragraph{Sampling protocol.} We evaluate every model by \emph{free sampling}: from the trained
model we draw $Q$ samples and keep all of them ($Q=10^4$ at $N=16$; $Q=10^5$ at $N=20$ and
$N=30$; for the genomic task $Q$ equals the number of distinct sequences $|S|$, $3{,}980$ at
$N{=}20$ and $2{,}716$ at $N{=}16$). Coverage is the fraction of the unseen valid set
$S\setminus T$ that appears among the distinct valid samples, and the fidelity $F$ the fraction of
new samples that are valid~\cite{gili2022}. A model that samples uniform random bits over all
$2^N$ strings reaches expected coverage $C_{\mathrm{rand}}=1-(1-2^{-N})^{Q}$; this no-learning
floor is our baseline. When
$|S|\gg Q$ the raw coverage is capped by the budget, so we also report the normalized coverage
$\tilde C=C/C^*$, where $C^*=1-(1-1/|S\setminus T|)^{Q}$ is the coverage an ideal uniform sampler
over $S\setminus T$ reaches in $Q$ draws. In this budget-limited regime every valid draw
is almost surely new and distinct, so $\tilde C\to E\,F$: at $N=30$ the $F$ and $\tilde C$
columns of Table~\ref{tab:benchmark} coincide and fidelity is the effective discriminator,
whereas at $N=16$ and $N=20$, where $Q$ is comparable to $|S|$, coverage carries information
beyond fidelity (Appendix~\ref{app:extra}). Because $p^*$ is known by design, we also compute the
forward KL $\mathrm{KL}(p^*\,\|\,q_{\theta})$ of the trained model.

\subsection{Cardinality-constrained data}
\label{sec:res_card}

After training, we measure both the MMD$^2$ and the coverage of every model. On this task the two
decouple. The deploy-quantum moment-trained models (the IQP Born machine and the active and magic
FBMs) converge the MMD$^2$ to a low value, yet do not generate valid bitstrings efficiently: a low
loss does not make them generators. The likelihood-trained models reach a comparably low MMD$^2$ and
also cover most of the valid set (Table~\ref{tab:benchmark}). One moment-trained model is the
exception. The passive number-conserving fermionic Born machine reaches high coverage because its
matchgate circuit produces only Hamming-weight-preserving strings, so it is confined to the valid
set by construction rather than by its loss. No other model was biased this way; when the valid set
has known structure, encoding it into the model raises coverage where lowering the loss does not.

At $N=16$, where each model's exact probabilities are available, we rank the models by coverage, by
the forward KL divergence, and by the eval MMD$^2$, and compare the rankings pairwise
(Fig.~\ref{fig:rank}, top row). Coverage and the forward KL agree closely ($r=0.93$). The
MMD$^2$ separates the seven autoregressive and tensor-network models from the six others,
which gives it rank correlations of $0.79$ with both generalization measures, and it fails to
order the models within each group: the IQP and tensor-network Born machines reach MMD$^2$ values
within a factor of $1.6$, yet cover $0.17$ and $0.52$ of the unseen valid set, and the model with
the lowest MMD$^2$, the small transformer, ranks seventh in coverage.
The pattern holds
for all evaluated models and across sizes to $N=30$ (Appendix~\ref{app:extra}).

Figure~\ref{fig:mechanism} shows the decoupling directly. During training the IQP Born
machine lowers its MMD$^2$ while its coverage stays flat at the random-bit coverage floor its
initialization starts it from (a), whereas a likelihood-trained transformer lowers its loss and raises its coverage together,
toward the ideal finite-sample ceiling $C^*$ (b). Capacity does not close the gap: enlarging the IQP
circuit $18\times$ in gate count lowers its coverage from $0.17$ to $0.14$ while its eval
MMD$^2$ stays at $1.9\times10^{-3}$, and the magic FBM stays at $0.26$--$0.28$ from one to four
FLO layers, while the tensor-network and classical
models improve with size (c; Table~\ref{tab:capacity}). The gap widens with qubit count, out to
$N=30$ (d). The decoupling is not an artifact of the kernel bandwidth: sweeping $\sigma$ from sharp
to broad leaves IQP near $C_{\mathrm{rand}}$ (Appendix~\ref{app:sigma}). The training
objective contributes to the gap. Training the same IQP circuit by exact likelihood, with every
other setting fixed, raises coverage from $0.17$ to $0.20$ and lowers the exact forward KL from
$1.81$ to $1.28$ at $N=16$, and raises coverage from $0.09$ to $0.14$ at $N=20$
(Appendix~\ref{app:init}). Most of the gap to the likelihood-trained models at $\approx0.5$
persists under exact likelihood training. Within the settings we tested, the objective
therefore accounts for a smaller part of the gap than the circuit itself; its generators are
two-body, and the constraint lives in correlators of every order (Section~\ref{sec:theory}).

\subsection{Genomic variant data}
\label{sec:res_genomic}

On the genomic data the IQP Born machine comes closer to the classical generators than on
the cardinality task. The medium and large transformers and RNNs and the tensor-network Born
machine reach the highest coverage, and the IQP Born machine reaches $C=0.112$ with the
fifth-lowest MMD$^2$ of the thirteen models (Table~\ref{tab:benchmark}). The classical GMMN,
trained on the same MMD$^2$ loss, comes within $0.01$ of the strongest likelihood-trained models
here but stays far below them on the cardinality task: a target dominated by low-order
statistics is within reach of a moment loss, whereas the cardinality constraint is a property of
correlators of every order (Section~\ref{sec:theory}). The loss still does not identify the
better generator. At $N=16$ the two lowest-loss models have the lowest coverage among the nine
models with $C>0.13$, while the highest-coverage model has the eighth-lowest loss
(Fig.~\ref{fig:rank}, bottom right). The fermionic Born machines reach at most $C=0.05$: their
output is confined to even Hamming weights, half of the genomic valid set
(Appendix~\ref{app:magic}), and the passive FBM, confined to a single Hamming weight,
generalizes poorly on a valid set that spans weights $3$ to $19$. Genomics-specific diagnostics (per-locus allele
frequencies, pairwise linkage disequilibrium, and a PCA of the generated samples) are reported in
Appendix~\ref{app:genomic}.

\section{Conclusion}
\label{sec:conclusion}

The train-classical, deploy-quantum paradigm arose from a practical constraint. Training a Born machine on hardware is expensive, and barren plateaus can make its gradients vanish, whereas the losses of several circuit families can be evaluated efficiently on a classical computer even when sampling their output is believed to be classically hard. Work on the paradigm has therefore asked whether the classical loss can be optimized and whether the deployed circuit can be replaced by a classical surrogate. The property that justifies a generative model in the first place, that it produces new and valid samples, had not been tested, because the paradigm by construction never samples the model during training.

We tested it, numerically and theoretically. Numerically, we sampled classical and quantum generative models trained on datasets whose valid set is known, so that generalization could be scored directly rather than inferred from the loss, at sizes reached by full state-vector simulation. Theoretically, we showed that for any loss fixed by a prescribed set of low-order correlators, exact matching is consistent with almost any coverage. Two distributions can drive such a loss to zero while one covers the whole valid set and the other an exponentially small fraction of it; the difference lies in the high-order correlators that the loss never constrains.

The two demonstrations agree. A converged moment-matching loss did not order the models by how well they generalized, and circuits trained through such a loss reached losses comparable with likelihood-trained models while generating far fewer new valid samples. The exception was a circuit that enforced the constraint of the task by construction, which generalized because of its structure rather than its loss. The implication for the paradigm is that trainability and sampling hardness, the two properties that motivated it, do not by themselves make a trained circuit a useful generator. Generalization has to be measured by sampling the trained model, and this measurement is not free. At the sizes where the paradigm is meant to pay off, sampling the trained circuit is the step believed to be classically intractable, so the quantum device is needed to certify the model before it is deployed. In a real application the valid set may itself be uncertain, and certification then also requires a reliable validity oracle.

If the device is required to certify the model, it can also inform the training. A workflow that trains classically, deploys on hardware, and returns a signal from the deployed samples to the classical optimization, as in the synergistic pretraining of Ref.~\cite{rudolph2023synergistic}, makes the quantum step part of training rather than a final readout. We see this closed loop as the direction the paradigm should take. Two further routes reduce how much the loss is asked to carry: losses that target generalization more directly, where they remain classically tractable, and models that build the constraint of the application into the circuit, as tensor networks that embed discrete linear constraints~\cite{lopezpiqueres2023symmetric,lopezpiqueres2025cons} and Hamming-weight-preserving circuits~\cite{monbroussou2025hwpreserving} do. Finding such structure for other domains, such as molecular validity or graph properties, remains open. In every case generalization is a property to be measured, not assumed from a low loss.

\begin{acknowledgments}
We thank Joseph Bowles, Zolt\'an Kolarovszki, and the Budapest Quantum Computing
Group for helpful correspondence on the IQP and fermionic Born machines.
\end{acknowledgments}

\bibliographystyle{apsrev4-2}
\bibliography{references}

\clearpage

\onecolumngrid
\appendix

\begin{center}
{\textbf{\Large Appendix}}\\[0.5em]
{\textbf{\large Table of Contents}}
\end{center}
{\small
\noindent Appendix~\ref{app:theory}\quad Proof of Theorem~\ref{thm:main}\dotfill\pageref{app:theory}\par
\noindent Appendix~\ref{app:efrc}\quad EFRC metrics and coverage normalization\dotfill\pageref{app:efrc}\par
\noindent Appendix~\ref{app:hyper}\quad Models and training details\dotfill\pageref{app:hyper}\par
\noindent Appendix~\ref{app:init}\quad Initialization of the IQP Born machine\dotfill\pageref{app:init}\par
\noindent Appendix~\ref{app:magic}\quad Configuration of the magic FBM\dotfill\pageref{app:magic}\par
\noindent Appendix~\ref{app:extra}\quad Additional cardinality-task results\dotfill\pageref{app:extra}\par
\noindent Appendix~\ref{app:sigma}\quad Kernel-bandwidth robustness\dotfill\pageref{app:sigma}\par
\noindent Appendix~\ref{app:genomic}\quad Genomic-task diagnostics\dotfill\pageref{app:genomic}\par
}
\vspace{1em}

\section{Proof of Theorem~\ref{thm:main}}
\label{app:proof}
\label{app:theory}

\begin{thmformal}
Let $S\subseteq\{0,1\}^N$ be a valid set, let $p^*$ be a target with $\supp(p^*)=S$, and let
$O_1,\dots,O_d:\{0,1\}^N\to[-1,1]$ be any bounded observables with $d+1<|S|$.
Let $\mathcal{L}$ be a loss that depends on the model only through the expectations
$\langle O_1\rangle_{q},\dots,\langle O_d\rangle_{q}$ and is minimized exactly when all of them
match $p^*$. Then:
\begin{itemize}
\item[(i)] \emph{Sparse exact minimizers:} There is a distribution $q$ with
$\supp(q)\subseteq S$ and $|\supp(q)|\le d+1$ that matches $p^*$ on every
$\langle O_j\rangle$; when the constant function lies in the span of the $O_j$, the count
improves to $d$, and every bound below holds with $d$ in place of $d+1$. The distribution
$q$ is therefore a global minimizer of $\mathcal{L}$, its support coverage obeys
$C_\infty(q)\le (d+1)/|S|$, and $\mathrm{KL}(p^*\,\|\,q)=\infty$, while the target
$p^*$ is also a global minimizer with $C_\infty(p^*)=1$. Hence a converged loss coexists
with $\widehat C_Q\le (d+1)/|S|$ at every $Q$.
\item[(ii)] \emph{Full-support exact minimizers:} For $0<\eta<1$ let
$q_\eta=(1-\eta)\,q+\eta\,p^*$ with $q$ the minimizer of (i). Then $q_\eta$ has full support
on $S$, so $C_\infty(q_\eta)=1$, it matches $p^*$ on every $\langle O_j\rangle$ and remains a global
minimizer, and yet
\begin{equation}
\label{eq:smoothkl}
\mathrm{KL}(p^*\,\|\,q_\eta)\;\ge\;\alpha\,\log\frac1\eta-\frac1e,
\qquad \alpha=p^*\big(S\setminus\supp(q)\big)\;\ge\;1-\max_{|V|\le d+1}p^*(V)\,.
\end{equation}
Choosing $\eta\le\exp[-(R+1/e)/\alpha]$ makes the forward KL exceed any prescribed finite $R$.
An exactly converged loss is therefore consistent with every value of the forward KL, finite
or infinite.
\end{itemize}
\end{thmformal}

\begin{proof}
\emph{(i)} We work in the coordinates $q(x)$ for $x\in S$ only, so that every candidate is
supported on the valid set by construction. We start at $q=p^*$, which matches all $d$
expectations, and we shrink its support one string at a time. Suppose the current $q$ has
support $\mathcal{R}$ with $|\mathcal{R}|>d+1$. The $d$ maps
$v\mapsto\sum_{x\in \mathcal{R}}v(x)O_j(x)$ and the mass
map $v\mapsto\sum_{x\in \mathcal{R}}v(x)$ together define a linear map
$\mathbb{R}^{|\mathcal{R}|}\to\mathbb{R}^{d+1}$, and since $|\mathcal{R}|>d+1$ its kernel
contains some $v\ne0$; every such $v$ satisfies $\sum_x v(x)=0$ by the last coordinate. We
replace $q$ by $q+tv$. All $d$ expectations are unchanged and the total mass stays
$1$, so $q+tv$ still matches $p^*$ on every $O_j$. Because the coordinates of $v$ sum to
zero, $v$ has a coordinate of each sign, and $q>0$ on $\mathcal{R}$, so we can grow $|t|$ in
the direction that decreases some coordinate until the first coordinate of $q+tv$ reaches
zero, which happens at a finite $t$. The support has strictly shrunk, and we repeat the
step. The process terminates, and it can stop only once the support holds at most $d+1$
strings. When the constant function lies in the span of the $O_j$, the mass map is a linear
combination of the $d$ maps, the joint system has rank at most $d$, and the same loop runs
until the support holds at most $d$ strings. The resulting $q$ is supported on at most $d+1$
strings of $S$ and matches $p^*$ on every $\langle O_j\rangle$, so it attains the minimum
of $\mathcal{L}$. Its support coverage obeys
$C_\infty(q)\le|\supp(q)|/|S|\le(d+1)/|S|$, and every distinct sampled string lies in
$\supp(q)$, so $\widehat C_Q(q)\le C_\infty(q)\le(d+1)/|S|$ at every $Q$. The target has
full support on $S$, so it too is a global minimizer, with $C_\infty(p^*)=1$; the pair
$(p^*,q)$ takes the same loss value, which is the statement of the informal theorem.

\emph{(ii)} Mixtures are linear in expectations, so
$\langle O_j\rangle_{q_\eta}=(1-\eta)\langle O_j\rangle_{q}+\eta\langle O_j\rangle_{p^*}
=\langle O_j\rangle_{p^*}$ for every $j$: $q_\eta$ remains a global minimizer, and it has full
support on $S$ because $p^*$ does. For the KL bound, we write $G=S\setminus\supp(q)$ and
$\alpha=p^*(G)$. Then $q_\eta(G)=\eta\,p^*(G)=\eta\alpha$, and the data-processing inequality
applied to the binary partition $\{G,G^c\}$ gives
\[
\mathrm{KL}(p^*\,\|\,q_\eta)\;\ge\;
\alpha\log\frac{\alpha}{\eta\alpha}+(1-\alpha)\log\frac{1-\alpha}{1-\eta\alpha}
\;\ge\;\alpha\log\frac1\eta+(1-\alpha)\log(1-\alpha)
\;\ge\;\alpha\log\frac1\eta-\frac1e,
\]
using $\log(1-\eta\alpha)<0$ and $t\log t\ge-1/e$ on $(0,1]$. Finally,
$\alpha\ge1-\max_{|V|\le d+1}p^*(V)$ because $\supp(q)$ is a set of at most $d+1$
points, and the choice
$\eta\le\exp[-(R+1/e)/\alpha]$ rearranges to $\alpha\log(1/\eta)-1/e\ge R$.

\end{proof}

\noindent The proof of (i) is a support-reduction loop and produces the minimizer inside $S$
directly; the linear program of Fig.~\ref{fig:lp} constructs the same kind of point.
The support count is the classical one: a distribution constrained on $d$ functionals
is matched by one on at most $d+1$ points (Carath\'eodory~\cite{caratheodory1911};
Tchakaloff~\cite{tchakaloff1957,bayerteichmann2006}; Richter~\cite{richter1957}).
The main text takes $T=\varnothing$, so the unseen set is $S$ itself. For a nonempty
training set the unseen set has $|S|-|T|$ strings, with $|T|$ the number of distinct training
strings; each bound becomes $(d+1)/(|S|-|T|)$, and
the change has relative size $O(|T|/|S|)$.

We work one example in full, taking $N=12$ and $L=2$. Then $D_L=1+12+66=79$ and
$|S|=\binom{12}{6}=924$, so the bound of Corollary~\ref{cor:covbound} is
$79/924\approx0.09$. The construction of Table~\ref{tab:worked} reaches support $66$ and
coverage $0.07$. For $L=3$, $D_L=299$ gives the bound $0.32$ and the construction reaches
$0.24$.

\begin{table}[t]
\centering
\small
\begin{tabular}{lccccc}
\toprule
model & $\MMD^2_{\sigma,\le L}$ & $\MMD^2$ (full) & $\mathrm{TV}(p^*,q)$ & $\mathrm{KL}(p^*\,\|\,q)$ & coverage\\
\midrule
target (uniform on $S$)      & $0$               & $0$                & $0$    & $0$      & $1.00$\\
sparse, $L=2$ (support $66$)  & $1\times10^{-28}$ & $9.6\times10^{-3}$ & $0.93$ & $\infty$ & $0.07$\\
sparse, $L=3$ (support $220$) & $5\times10^{-25}$ & $1.5\times10^{-3}$ & $0.78$ & $\infty$ & $0.24$\\
\bottomrule
\end{tabular}
\caption{The uniform target and two sparse exact minimizers of the truncated loss at $N=12$. The truncated loss is numerically zero
for every row, while coverage and the KL separate the rows. The full MMD is nonzero at the low-coverage rows, consistent with its characteristic
kernel.}
\label{tab:worked}
\end{table}

\section{EFRC metrics and coverage normalization}
\label{app:efrc}
This appendix collects the sample-based metrics of Gili et al.~\cite{gili2022} in the
notation of the main text. As in Section~\ref{sec:framework}, $S\subseteq\{0,1\}^N$ is the
valid set, $T$ the training set, and $U=S\setminus T$ the set of unseen valid strings; the
training fraction is $\varepsilon=|T\cap S|/|S|$, so $|U|=|S|(1-\varepsilon)$. From the
trained model we draw $Q$ samples $Y_1,\dots,Y_Q$. Among these, let $G_{\mathrm{new}}$
denote the samples that lie outside $T$, counted with multiplicity; let $G_{\mathrm{sol}}$
denote the samples that lie in $U$, that is, are both new and valid, again with
multiplicity; and let $g_{\mathrm{sol}}$ denote the set of \emph{distinct} strings in
$G_{\mathrm{sol}}$. The four EFRC metrics are
\begin{equation}
\label{eq:efrc}
E=\frac{|G_{\mathrm{new}}|}{Q},\qquad
F=\frac{|G_{\mathrm{sol}}|}{|G_{\mathrm{new}}|},\qquad
R=\frac{|G_{\mathrm{sol}}|}{Q}=E\,F,\qquad
C=\frac{|g_{\mathrm{sol}}|}{|U|}.
\end{equation}
The exploration $E$ is the fraction of samples that leave the training set; the fidelity
$F$ is the fraction of those new samples that are valid; the rate $R$ is the fraction of
all samples that are new and valid; and the coverage $C$ is the fraction of the unseen
valid set that the $Q$ samples reach, which is the measured coverage $\widehat C_Q$ of
Section~\ref{sec:framework}.

Two of the metrics are reported in normalized form. The normalized rate is
$\tilde R = R/(1-\varepsilon)$, which compares $R$ to the value $1-\varepsilon$ that an
exact sampler of the uniform target attains in expectation. The normalized coverage is
$\tilde C = C/C^*$, where $C^* = 1-(1-1/|U|)^{Q}$ is the expected coverage that an ideal
uniform sampler over $U$ attains after $Q$ draws. Under i.i.d.\ sampling from a
distribution supported on $U$, the uniform distribution maximizes expected coverage, so its
expected normalized coverage is $1$; observed $\tilde C$ can fluctuate above $1$.

\section{Models and training details}
\label{app:hyper}
All models train on the same training set $T$ at each $(\varepsilon,\text{seed})$ pair and
are evaluated by the free-sampling protocol of Section~\ref{sec:setup}, with three seeds for
the cardinality task and five for the genomic task. The circuit families are implemented in
JAX, the neural models in PyTorch, and the RBM wraps scikit-learn; every model trains on
CPU, and the gradient-based models use the Adam optimizer.

\paragraph{Likelihood-trained models.} The TNBM is the matrix-product-state Born machine of
Han et al.~\cite{han2018mps}, trained by two-site DMRG on the negative log-likelihood with
an adaptive bond dimension: each local update takes $10$ gradient substeps of size $0.01$,
singular values below $10^{-2}$ are truncated, the bond dimension is capped at $\alpha=7$,
and training runs for up to $100$ sweeps. The transformers are decoder-only autoregressive
models over the $N$ bits at three sizes, with $d_{\rm model}=16/32/64$, $2/4/4$ attention
heads, $1/2/3$ layers, and dropout $0.1$, giving $3.4$K$/25.6$K$/150$K parameters. The RNNs
are autoregressive GRUs with hidden dimension $16/64/128$ and $1/1/2$ layers
($1.7$K$/25$K$/199$K parameters). All autoregressive models minimize the cross-entropy
(negative log-likelihood) at learning rate $10^{-3}$ with batch size $64/128/256$
(small/medium/large) for up to $200$ epochs, with early stopping at patience $20$. The RBM
has $24$ hidden units and trains by persistent contrastive divergence (learning rate
$0.01$, batch size $64$, $200$ iterations); its samples are drawn by $1{,}000$ steps of
block Gibbs sampling from random initial states.

\paragraph{Moment-trained models.} The IQP Born machine follows
Ref.~\cite{recio2025iqpopt}: a parametrized IQP circuit with all-to-all connectivity, one
angle per one- and two-body Pauli-$Z$ generator ($N+\binom{N}{2}$ parameters; the
generators commute, so depth adds no parameters), trained on the MMD$^2$ with the Gaussian--Hamming kernel at $\sigma=1$, estimated by
the sampled-spectrum den Nest estimator with $100$ Pauli-$Z$ words and $256$ samples per
step, for $500$ steps at learning rate $0.01$. Parameter counts at the sizes of
Table~\ref{tab:benchmark} ($N=30$ cardinality / $N=20$ genomic): transformers
$3{,}394$/$25{,}634$/$150{,}402$ and RNNs $1{,}714$/$25{,}282$/$198{,}786$ at every $N$;
TNBM $2{,}212$/$1{,}232$; RBM $774$/$524$; GMMN $77{,}854$/$75{,}284$; IQP $465$/$210$;
passive FBM $435$/$190$; active FBM $1{,}540$/$780$; magic FBM $1{,}547$/$785$ (the
FBM cardinality entries are at $N=28$, Appendix~\ref{app:magic}).
The passive (number-conserving) FBM
parametrizes $O=e^{W}\in SO(N)$ with $N(N{-}1)/2$ parameters ($120$ at $N=16$, $190$ at
$N=20$); its correlation matrix $C=O_{[:,:k]}O_{[:,:k]}^{\top}$ is a rank-$k$ projector, so
the Born distribution lives entirely in the $\mathrm{HW}{=}k$ sector and $F=1$ holds by
construction. It trains on the mean-square error of all Pauli-$Z$ string correlators of
order at most $3$ against their training-set values, at learning rate $0.05$ for $1{,}500$
steps. The active FBM is the compound-matchgate model of Ref.~\cite{bako2025fbm} ($496$
parameters at $N=16$), trained on the same correlator loss with order at most $5$, at
learning rate $10^{-3}$ for $2{,}000$ steps; at $N=28$ the loss uses order at most $3$
and $1{,}500$ steps. Its state is the vacuum evolved by an $SO(2N)$ matchgate circuit, a
fermionic Gaussian state of even parity, and it is sampled exactly: from the $2^N$ Born vector
at $N\le20$ and by sequential projective Gaussian measurements above. The magic FBM augments the matchgate circuit
with magic input states (register configurations in Appendix~\ref{app:magic}) and trains on
a stochastic estimate of the locality-$\le3$ truncated $\sigma=1$ MMD$^2$: at each step it
draws $200$ Pauli-$Z$ strings from the kernel's locality distribution and compares their
exact model expectations with the training-set values, at learning rate $10^{-3}$ for $2{,}000$ steps; its
state is sampled exactly from the statevector. Both FBMs prepare even-parity states and
preserve parity, so their samples have even Hamming weight. The GMMN~\cite{li2015gmmn,dziugaite2015}
is a multilayer perceptron that maps a $16$-dimensional Gaussian latent vector through two
hidden layers of $256$ units to $N$ logits, made differentiable by the binary-concrete
(Gumbel--sigmoid) relaxation at temperature $\tau=0.5$; it trains on the $\sigma=1$
Gaussian--Hamming MMD$^2$ between batches of $256$ generated and $256$ training samples, at
learning rate $10^{-3}$ for $5{,}000$ steps.

\section{\texorpdfstring{Initialization of the IQP Born machine}{Initialization of the IQP Born machine}}
\label{app:init}
{

The IQP results are sensitive to how the single-body angles are initialized, and we report
two data-dependent initialization strategies in the spirit of
Ref.~\cite{lerch2026iqpinit}. Both set the single-body angles from the training-data
magnetizations $m_j=\langle(-1)^{x_j}\rangle_{T}$ and draw the two-body angles from
$\mathcal{N}(0,0.1^2)$; they differ in the angle map. The \emph{magnetization} strategy uses
the Boltzmann-type map $\theta_j=\operatorname{arctanh}(m_j)/2$. Under the IQP
parameterization the initial single-bit marginal is $\langle Z_j\rangle=\cos 2\theta_j$, so
the map that matches the training marginals is instead $\theta_j=\arccos(m_j)/2$, the
\emph{marginal-matching} strategy. The two maps differ most where it matters:
on the genomic data ($13$ of $20$ loci are majority-one) the arctanh map leaves the initial
marginals near $\langle Z_j\rangle\approx+1$ on essentially every locus, so training starts
close to the $|0\ldots0\rangle$-biased regime whose poor trainability is analyzed in
Ref.~\cite{lerch2026iqpinit}; on the cardinality target the magnetizations vanish and the
map reduces to a small random initialization with the same bias, whereas the
marginal-matching value is $\theta_j=\pi/4$, which starts the model at the random-bit
coverage floor.

We train the IQP model from each initialization with every other setting of the protocol
identical. Table~\ref{tab:init} reports the outcome. On the genomic task the initialization
accounts for most of the gap to the classical models: coverage rises from $0.014$ to
$0.112$, against $0.13$--$0.14$ for the strongest classical models, and the eval MMD$^2$
falls to one of the lowest values in the benchmark. On the cardinality task the choice
matters far less: coverage moves from $0.14$ to $0.17$, just above the untrained
random-bit coverage floor of $0.14$ and well below the likelihood-trained models at
$\approx0.5$, so the cardinality gap is not an initialization artifact.

\begin{table}[!ht]
\centering
\small
\begin{tabular}{llcc}
\toprule
task & initialization & MMD$^2$ ($\sigma{=}1$) & $C$\\
\midrule
genomic ($N{=}20$, $\varepsilon{=}0.2$) & magnetization (arctanh) & $1.08\times10^{-2}$ & $0.014$\\
                                        & marginal-matching (arccos) & $1.57\times10^{-3}$ & $0.112$\\
\midrule
cardinality ($N{=}16$, $\varepsilon{=}0.1$) & magnetization (arctanh) & $1.91\times10^{-3}$ & $0.136$\\
                                            & marginal-matching ($\theta_j{=}\pi/4$) & $1.86\times10^{-3}$ & $0.170$\\
                                            & $\theta_j{=}\pi/4$, untrained (one seed) & --- & $0.142$\\
\bottomrule
\end{tabular}
\caption{IQP under the two single-body initialization maps, with all other protocol settings
fixed (medians over five seeds for the genomic task and three for the cardinality task;
per-seed coverage $0.165$/$0.178$/$0.170$ for the marginal-matching cardinality row). The
untrained $\pi/4$ row is the random-bit coverage floor. The strongest classical models reach
$C=0.13$--$0.14$ on the genomic task and $C\approx0.5$ on the cardinality task.}
\label{tab:init}
\end{table}

The initialization study above varies the starting point at a fixed loss; the
complementary ablation fixes the architecture and varies the loss. At $N\le20$ the IQP
model's Born probabilities are exactly computable, so the same circuit (same generators,
parameters, initialization, and optimizer) can be trained by the exact negative
log-likelihood. At $N=16$, $\varepsilon=0.10$ (three seeds, all
quantities computed exactly from the probability vector), likelihood training improves
every generalization measure: coverage $0.17\to0.20$, exact forward KL $1.81\to1.28$,
valid mass $0.26\to0.31$; at $N=20$ coverage moves from $0.09$ to $0.14$ and the exact
forward KL from $2.27$ to $1.36$ (three seeds at both sizes). The
improvements are modest: even trained by exact likelihood, the two-body IQP
circuit stays well below the classical models at $\approx0.5$. On this task the training
objective accounts for part of the gap; the rest persists when the same circuit is trained
by exact likelihood, with two-body generators against a constraint that lives in
correlators of every order. (Fig.~\ref{fig:rank} uses the sample-based, add-$\alpha$-smoothed KL estimator
uniformly across models, which penalizes thin support more strongly than the exact KL
reported here.)

The weight-$\le2$ generator set of the main tables is not the capacity ceiling of the
family. Extending the generators to all Pauli-$Z$ words of weight $\le4$ ($6{,}195$
parameters at $N{=}20$) and training from the marginal-matching initialization with the
higher-order angles near zero requires a gentler schedule ($3{,}000$ steps at learning rate
$10^{-3}$; at learning rate $0.01$ the sampled gradients destabilize the larger circuit and
genomic coverage collapses to $0.004$). Trained this way, the weight-$\le4$ model reaches
median genomic coverage $0.115$ and eval MMD$^2$ $1.2\times10^{-3}$ over five seeds; the
weight-$\le2$ model under the same $3{,}000$-step schedule reaches $0.122$ and
$1.27\times10^{-3}$ ($0.112$ and $1.57\times10^{-3}$ under the $500$-step schedule of the
main tables). The $30\times$ larger ansatz does not raise coverage. On the cardinality task
at $N=16$ the same sweep under the marginal-matching initialization gives
$C=0.170/0.134/0.136$ for weights $\le2/3/4$, with exact forward KL $1.81/2.21/2.21$ and
eval MMD$^2$ flat at $1.9\times10^{-3}$ (Table~\ref{tab:capacity}).

The initialization also affects the rank correlations of Fig.~\ref{fig:rank}. With the
magnetization initialization the correlations are $0.84/0.74/0.79$ (cardinality) and
$0.86/0.78/0.73$ (genomic) for the KL--coverage, KL--MMD$^2$, and coverage--MMD$^2$ pairs,
against $0.93/0.79/0.79$ and $0.85/0.63/0.53$ in Fig.~\ref{fig:rank}. The genomic MMD$^2$
pairings are weaker under the marginal-matching initialization because that model
attains the second-lowest MMD$^2$ at $N=16$ while its coverage and forward KL
remain mid-table: a converged loss again fails to identify the better generator. The $N=30$ entries and the
query-budget sweep of Fig.~\ref{fig:covq} use the magnetization initialization, for which the
two give the same normalized coverage ($\tilde C=0.14$), so the choice does not affect the
large-$N$ conclusion.

}

\section{\texorpdfstring{Configuration of the magic FBM}{Configuration of the magic FBM}}
\label{app:magic}
{

The magic FBM of Ref.~\cite{bako2025fbm} arranges the fermionic modes into registers, each
prepared in the magic input state $\cos\alpha\,|0\ldots0\rangle+\sin\alpha\,|1\ldots1\rangle$
and evolved by FLO layers; the magic angles $\alpha$ and the FLO Givens angles are optimized
jointly. The \emph{canonical} configuration of Ref.~\cite{bako2025fbm} uses four-mode
registers with three measured modes and one traced-out ancilla. Tracing the ancilla leaves a
non-Gaussian mixed state whose Born distribution we can sample exactly only at small
register counts, so we use the \emph{all-measured} configuration: four-mode registers with
every mode measured and no ancilla, $N/4$ registers, one FLO layer unless stated otherwise
($500$ parameters at $N=16$, $785$ at $N=20$). Every locus enters the loss and the samples
are drawn exactly from the $2^N$ Born vector. The configuration requires $N$ divisible by
four. Its output also has even parity (Appendix~\ref{app:hyper}), so at $N=30$ the target
weight $k=15$ would give $F=0$ by construction, as for the active FBM. The cardinality
entries of both FBMs in Table~\ref{tab:benchmark} and Fig.~\ref{fig:scatter} (double dagger)
are therefore taken at $N=28$, HW$=14$, the largest even-weight size at which the $2^N$ Born
vector of the magic FBM is still computed exactly, under the protocol of the $N=30$ column
($|T|=50{,}000$ training strings, $\varepsilon\approx1.2\times10^{-3}$, $Q=10^5$, three
seeds): the magic FBM with seven registers ($1{,}547$ parameters) and the active FBM with
locality-$3$ loss ($1{,}540$ parameters). The active FBM reaches $F=\tilde C=0.38$ at an
MMD$^2$ of $1.05\times10^{-3}$ and the magic FBM $F=\tilde C=0.31$ at $1.06\times10^{-3}$ (seed spread below $0.02$), while the likelihood-trained
models reach $\tilde C$ between $0.7$ and $1.0$ at $N=30$.

Because the output has even parity, on the cardinality task at $N=16$ and $N=20$ the target weight
is even ($k=8$, $10$), and the fidelity is set by the mass the model places on that weight
($0.44$ to $0.53$ at $N=16$). On the genomic task a fraction $0.51$ of the valid sequences
has even weight at $N=20$ ($0.49$ at $N=16$), which bounds the coverage of this
configuration by $0.46$ at $Q=3{,}980$; the bound does not bind at the values observed.

Two training settings are compared: ours, the stochastic locality-$\le3$ estimate of
Appendix~\ref{app:hyper} with $200$ Pauli-$Z$ strings per step, Givens angles initialized
from $\mathcal{N}(0,0.1^2)$, and $2{,}000$ steps; and the settings of
Ref.~\cite{bako2025fbm}, locality cutoff $5$, $256$ strings per step, Givens angles
initialized uniformly in $[0,1)$, and $1{,}000$ steps. Both use Adam at learning rate
$10^{-3}$, magic angles initialized uniformly in $[0,1)$, and $\sigma=1$.

\begin{table}[!ht]
\centering
\small
\begin{tabular}{llccc}
\toprule
task & training settings & MMD$^2$ ($\sigma{=}1$) & $F$ & $C$\\
\midrule
genomic ($N{=}20$, $\varepsilon{=}0.2$)
  & ours & $3.80\times10^{-3}$ & $0.063$ & $0.046$ ($0.036$--$0.052$)\\
  & Ref.~\cite{bako2025fbm} & $2.58\times10^{-3}$ & $0.075$ & $0.064$ ($0.054$--$0.070$)\\
\midrule
genomic ($N{=}16$, $\varepsilon{=}0.2$)
  & ours & $3.01\times10^{-3}$ & $0.227$ & $0.098$\\
  & Ref.~\cite{bako2025fbm} & $1.93\times10^{-3}$ & $0.272$ & $0.170$\\
\midrule
cardinality ($N{=}16$, $\varepsilon{=}0.1$)
  & ours, $L=1$ & $1.89\times10^{-3}$ & $0.442$ & $0.260$\\
  & ours, $L=2$ & $1.79\times10^{-3}$ & $0.482$ & $0.255$\\
  & ours, $L=4$ & $1.60\times10^{-3}$ & $0.528$ & $0.275$\\
\bottomrule
\end{tabular}
\caption{The all-measured magic FBM under the two training settings (medians over five
seeds for the genomic task and three for the cardinality task; per-seed ranges in
parentheses for the $N=20$ rows). $L$ is the number of FLO layers. The strongest classical
models reach $C=0.13$--$0.14$ on the genomic task at $N=20$, $0.25$--$0.27$ at $N=16$, and
$C\approx0.5$ on the cardinality task.}
\label{tab:magic}
\end{table}

Table~\ref{tab:magic} shows that the settings of Ref.~\cite{bako2025fbm} improve every
measure on the genomic task, where coverage rises from $0.046$ to $0.064$ at $N=20$ and
from $0.098$ to $0.170$ at $N=16$ and remains below the classical models. On the
cardinality task the depth of the FLO circuit changes coverage little ($0.26$ to $0.28$ from
one to four layers) while the fidelity rises from $0.44$ to $0.53$, and every variant stays
far below the likelihood-trained models, consistent with Section~\ref{sec:theory}. The
magic-FBM entries of Tables~\ref{tab:benchmark}, \ref{tab:capacity}, and the
per-$\varepsilon$ tables of Appendix~\ref{app:extra}, and the scores of
Fig.~\ref{fig:rank}, use our settings, which Table~\ref{tab:benchmark} applies uniformly
across models (Appendix~\ref{app:hyper}); the Ref.~\cite{bako2025fbm} rows are the
stronger-training reference for the architecture. Within this family the loss and the
coverage improve together across settings. The loss value nevertheless does not place the
family relative to the classical models: on the cardinality task the magic FBM reaches an
MMD$^2$ within a factor of two of the likelihood-trained models at a quarter of their
coverage, and on the genomic task its loss stays above theirs.

}

\section{Additional cardinality-task results}
\label{app:extra}
Table~\ref{tab:benchmark} reports the cardinality task at $N=30$ and the genomic variants at $N=20$;
the full per-$\varepsilon$ results below, at $N=16$ and $N=20$, list exploration $E$ and rate $R$
in addition to $F$ and $C$, and report both the $\sigma{=}1$ and median-heuristic MMD$^2$ kernels.

\begin{figure}[!ht]
\centering
\includegraphics[width=0.92\textwidth]{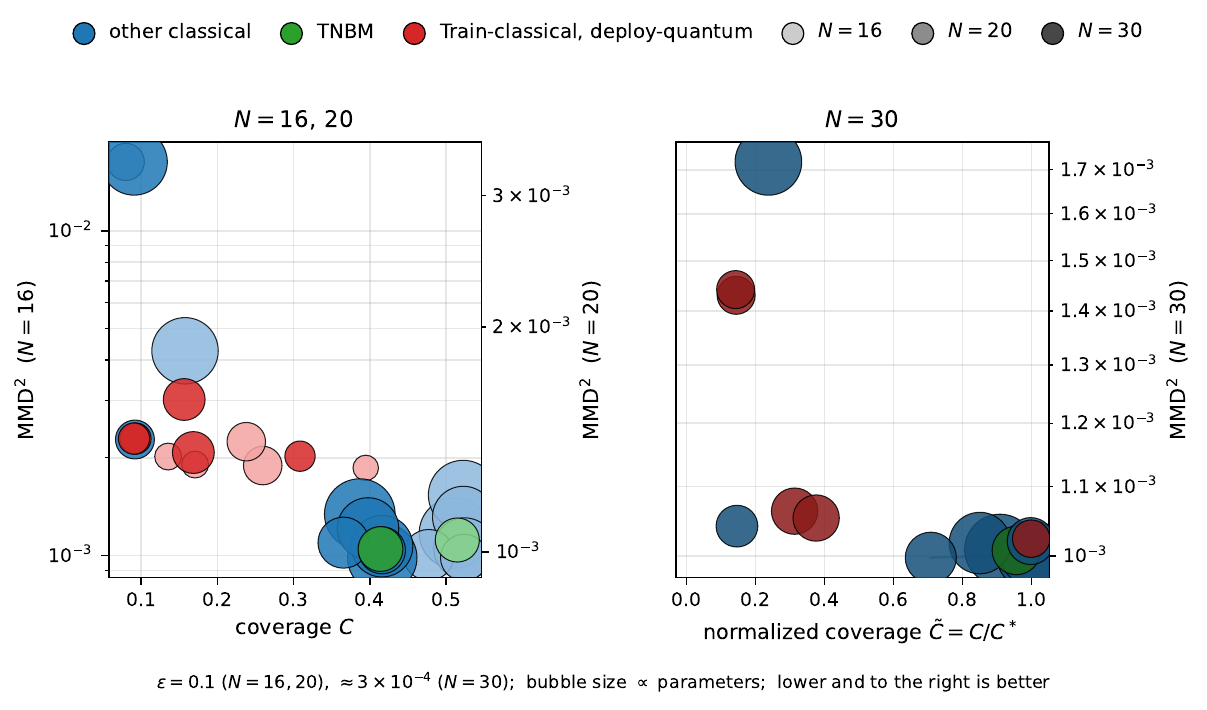}
\caption{\textbf{Training loss does not track coverage across the model zoo.} The same data as
Table~\ref{tab:benchmark}, shown as a scatter. \emph{Left:} the cardinality task at $N=16$ and $N=20$
($\varepsilon=0.10$, twin MMD$^2$ axes); \emph{right:} $N=30$ with normalized coverage
$\tilde C=C/C^*$. Each bubble is a model, colored by deployment group as in Fig.~\ref{fig:rank}
(red: train-classical, deploy-quantum; green: TNBM; blue: other classical), with
lighter-to-darker shade marking $N=16/20/30$ and size proportional to parameter count; lower and
to the right is better. The deploy-quantum moment-trained models reach only low coverage, the
exception being the structurally-constrained passive FBM on the cardinality task, yet their
MMD$^2$ spans the same range as the likelihood models that generalize. The active and magic FBM bubbles of the right panel are their $N{=}28$ entries of
Table~\ref{tab:benchmark} (double dagger, Appendix~\ref{app:magic}).}
\label{fig:scatter}
\end{figure}

\input{tables/master_n16.tex}
\input{tables/master_n20.tex}

Figure~\ref{fig:dynfull} shows the all-models training dynamics at both $N=16$ and $N=20$: every
model's MMD$^2$ converges to the same numerical range while coverage separates.

\begin{figure}[!ht]
\centering
\includegraphics[width=0.98\columnwidth]{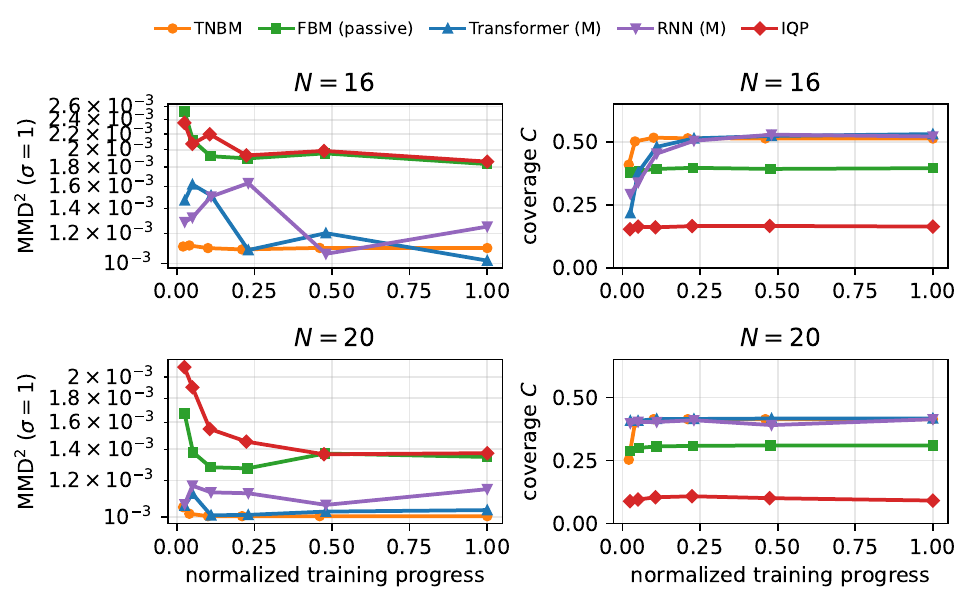}
\caption{\textbf{Loss converges while coverage separates.} Training trajectories at $N=16$ and $N=20$ for five models ($\varepsilon{=}0.10$; $N=16$
top, $N=20$ bottom). \emph{Left:} the MMD$^2$ converges to the same numerical range for every model. \emph{Right:}
coverage $C$ separates them: the TNBM, transformer, and RNN increase, the passive FBM
sits at its structural value, and IQP stays flat at the level its initialization starts it from.}
\label{fig:dynfull}
\end{figure}

The decoupling at the largest size is shown directly in Fig.~\ref{fig:n30}: every model's
MMD$^2$ converges to the same numerical range while the normalized coverage separates. In normalized coverage
the distance between the best model and IQP grows from $0.32$--$0.67$ at $N=8$--$16$ to $0.86$ at $N=30$
(main-text Fig.~\ref{fig:mechanism}d). Fig.~\ref{fig:n30loss} shows that at $N=30$ the sampled IQP
MMD$^2$ does not descend during training, whereas the transformer NLL converges with coverage near
the finite-sample maximum.

\begin{figure}[!ht]
\centering
\includegraphics[width=0.98\columnwidth]{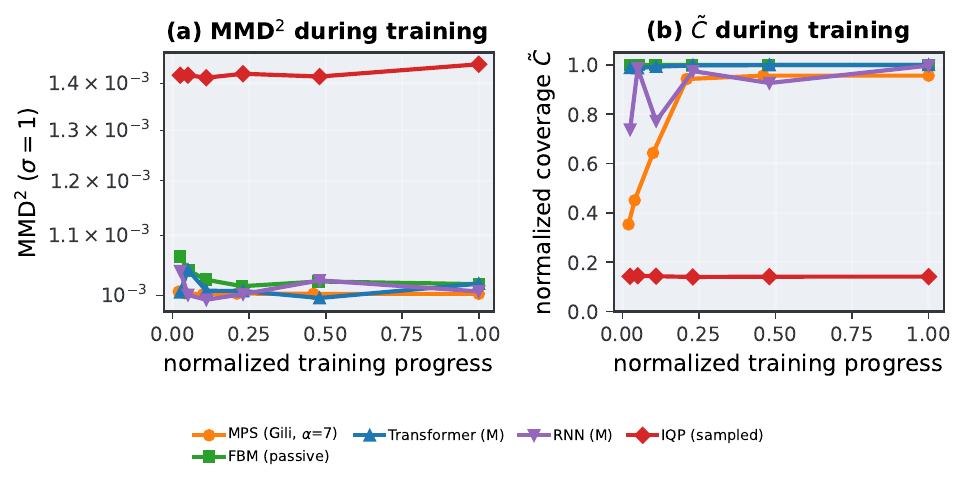}
\caption{\textbf{The decoupling persists at $N=30$.} Training trajectories at $N=30$, HW$=15$ ($|S|\approx1.55\times10^8$,
$\varepsilon\approx3.2\times10^{-4}$). \emph{(a)} every model's MMD$^2$ converges to the same
$\sim10^{-3}$ range. \emph{(b)} normalized coverage $\tilde C$ separates them: passive FBM,
transformer, and RNN approach $1$, the TNBM reaches $0.96$, and IQP stays near $C_{\mathrm{rand}}$.}
\label{fig:n30}
\end{figure}

\begin{figure}[!ht]
\centering
\includegraphics[width=0.98\columnwidth]{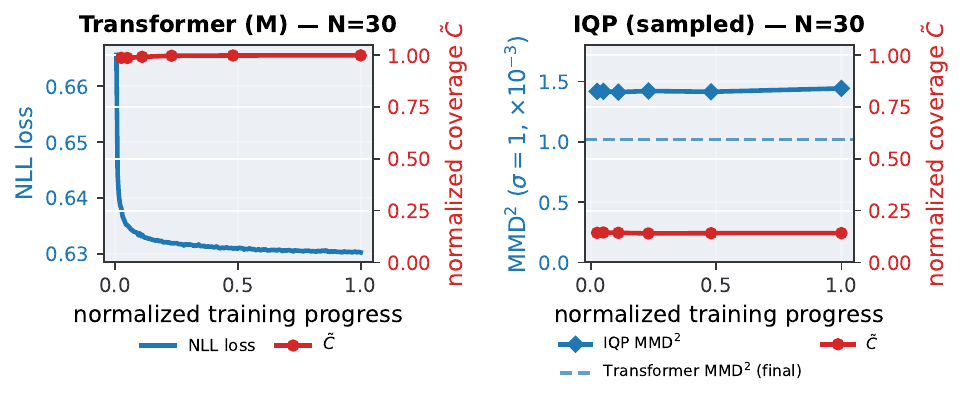}
\caption{\textbf{Loss-level view at $N=30$.} \emph{Left:} the transformer NLL falls and $\tilde C$ sits near the finite-sample maximum. \emph{Right:} the sampled IQP MMD$^2$ stays flat at $\approx1.4\times10^{-3}$
with $\tilde C$ near $C_{\mathrm{rand}}$, and never drops below the transformer's final MMD$^2$
(dashed): the NLL-trained model reaches a lower MMD$^2$ than the MMD-trained one.}
\label{fig:n30loss}
\end{figure}

\begin{figure}[!ht]
\centering
\includegraphics[width=0.98\columnwidth]{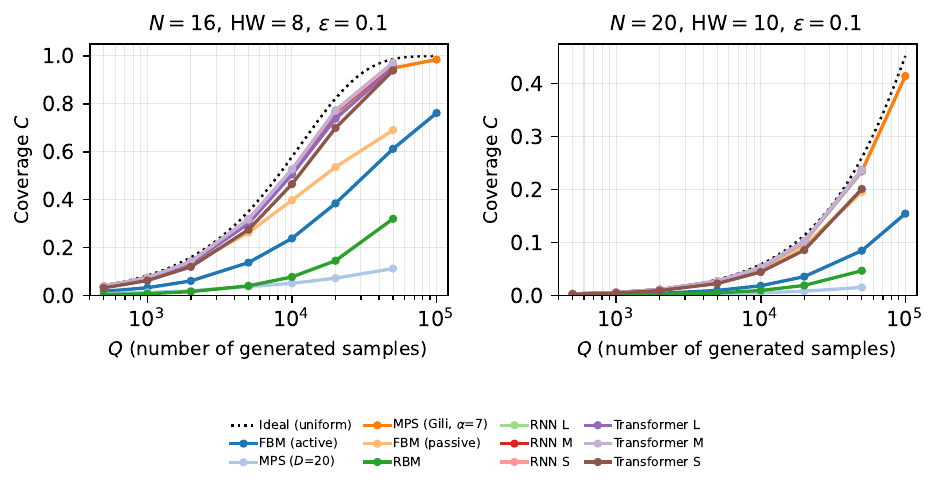}
\caption{\textbf{Low coverage is not a finite-sampling artifact.} Coverage versus query budget $Q$ at $N=16$ and $N=20$. The likelihood-trained and
structural models increase toward $C^*$ as $Q$ grows; the IQP Born machine
stays near $C_{\mathrm{rand}}$ across the budgets we study, and the active FBM, whose
Born mass on the valid weight is $0.40$ at $N=16$, rises more slowly than the
likelihood-trained models and reaches $0.76$ only at $Q=10^5$, ten times the budget of
the main tables. In the legend,
``MPS (Gili, $\alpha$=7)'' is the TNBM, and ``MPS ($D$=20)'' is a fixed-bond
($\chi=20$) MPS baseline.}
\label{fig:covq}
\end{figure}

\begin{figure}[!ht]
\centering
\includegraphics[width=\columnwidth]{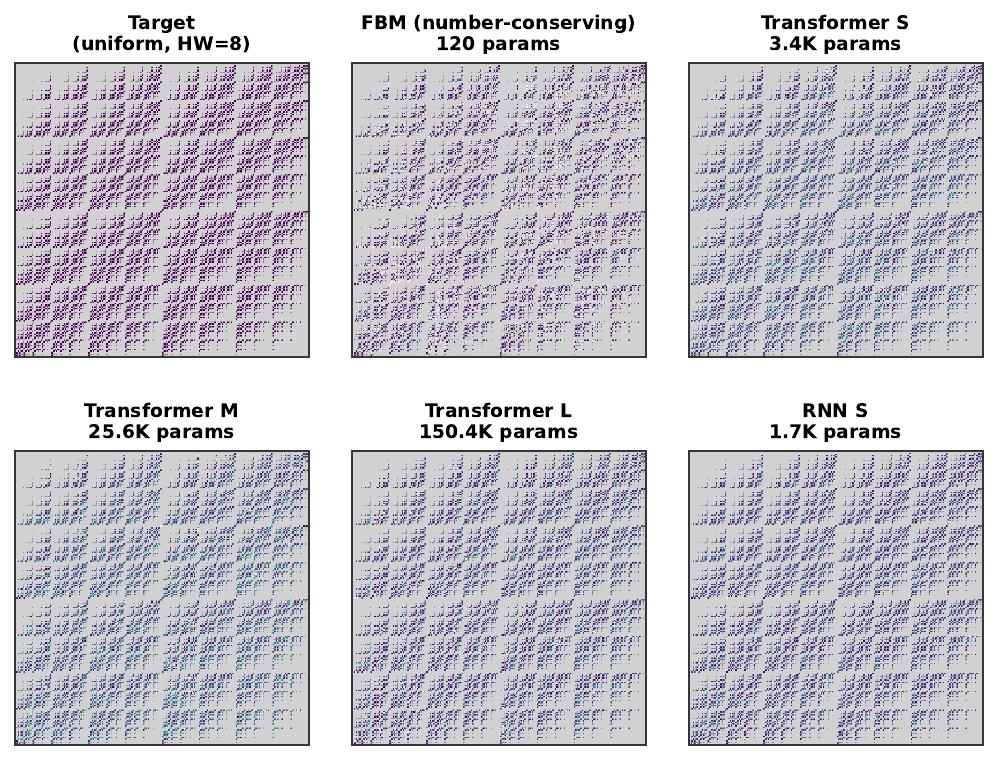}
\caption{\textbf{Where the models place their probability mass.} The full $2^{16}$ bitstring space at $N=16$, HW$=8$, rendered as a $256\times256$
grid; gray cells are off the HW$=8$ slice, colored cells are HW$=8$ with intensity
proportional to sample probability. The number-conserving FBM places all its mass on the
valid slice by construction; the moment-trained models assign substantial probability outside the valid sector.}
\label{fig:grid}
\end{figure}

\section{Kernel-bandwidth robustness}
\label{app:sigma}
IQP's low coverage is not an artifact of the kernel bandwidth. At $N=16$, HW$=8$, sweeping $\sigma$ from the
median heuristic of~\cite{recio2025iqpopt} to sharper and flatter values leaves the final
IQP coverage between $0.13$ and $0.17$ (peaking at $\sigma=2$), near the uniform Born
weight on the HW$=8$ sector and well below the likelihood-trained models at $\approx0.5$.
\begin{center}\small
\begin{tabular}{lccc}
\toprule
kernel bandwidth $\sigma$ & $F$ & $C$ & Born mass on HW$=k$\\
\midrule
median (Hamming, $=8$)    & 0.18 & 0.13 & 0.20\\
median (Euclidean, $=2.83$) & 0.18 & 0.14 & 0.20\\
$\sigma=1$                & 0.23 & 0.16 & 0.25\\
$\sigma=0.5$              & 0.18 & 0.13 & 0.20\\
\bottomrule
\end{tabular}
\end{center}
No bandwidth raises IQP's coverage anywhere near the likelihood-trained models; the
main tables use $\sigma=1$. The sweep is a single run of $800$ steps per bandwidth, which
is why its $\sigma=1$ coverage ($0.16$) differs from the three-seed, $500$-step median of
the main tables ($0.17$).

\begin{figure}[!ht]
\centering
\includegraphics[width=0.98\columnwidth]{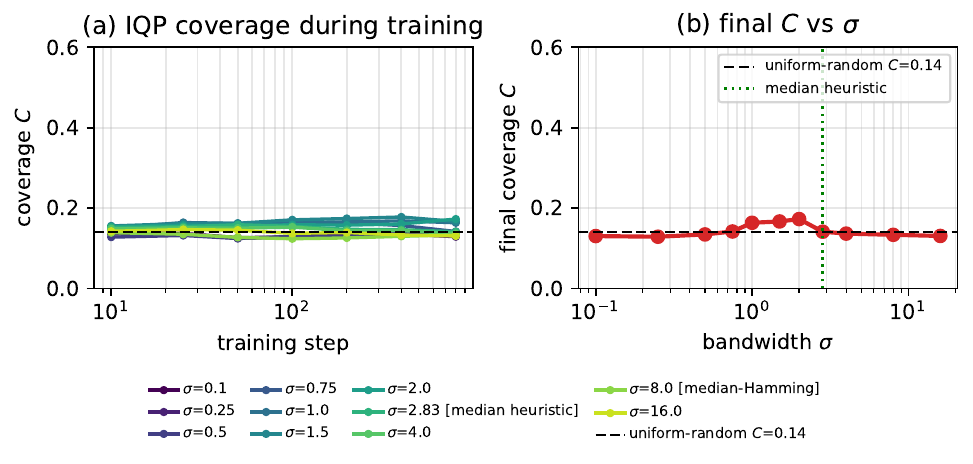}
\caption{\textbf{No kernel bandwidth rescues IQP coverage.} IQP coverage as the kernel
bandwidth $\sigma$ is swept from sharp ($0.1$) to broad ($16$) at $N=16$, including the median
heuristic (Euclidean, $\sigma=2.83$) and the median Hamming value ($\sigma=8$). \emph{(a)}
coverage during training for each $\sigma$; \emph{(b)} final coverage versus $\sigma$.
IQP stays at $0.13$ to $0.17$ for every
$\sigma$, near the random-bit baseline $C_{\mathrm{rand}}$ (dashed); the likelihood-trained
baselines of Table~\ref{tab:capacity} reach $\sim0.5$ on the same task.}
\label{fig:sigma}
\end{figure}

\section{Genomic-task diagnostics}
\label{app:genomic}
Beyond coverage, we check the generated genomic samples against standard population-genetics
statistics at $\varepsilon{=}20\%$. Figure~\ref{fig:af_ld} compares per-locus allele frequencies
of generated versus reference data; Fig.~\ref{fig:ldmat} shows the pairwise linkage-disequilibrium
(LD, $r^2$) matrices; Fig.~\ref{fig:hwdist} the Hamming-weight distributions; and
Fig.~\ref{fig:genpca} a PCA of the generated samples against the reference. The likelihood-trained
generators (transformer, RNN) and the TNBM reproduce the allele frequencies, the LD structure, and
the sample cloud closely, whereas the deploy-quantum moment-trained models depart from all three,
consistent with their low coverage.

\begin{figure}[!ht]
\centering
\includegraphics[width=0.95\columnwidth]{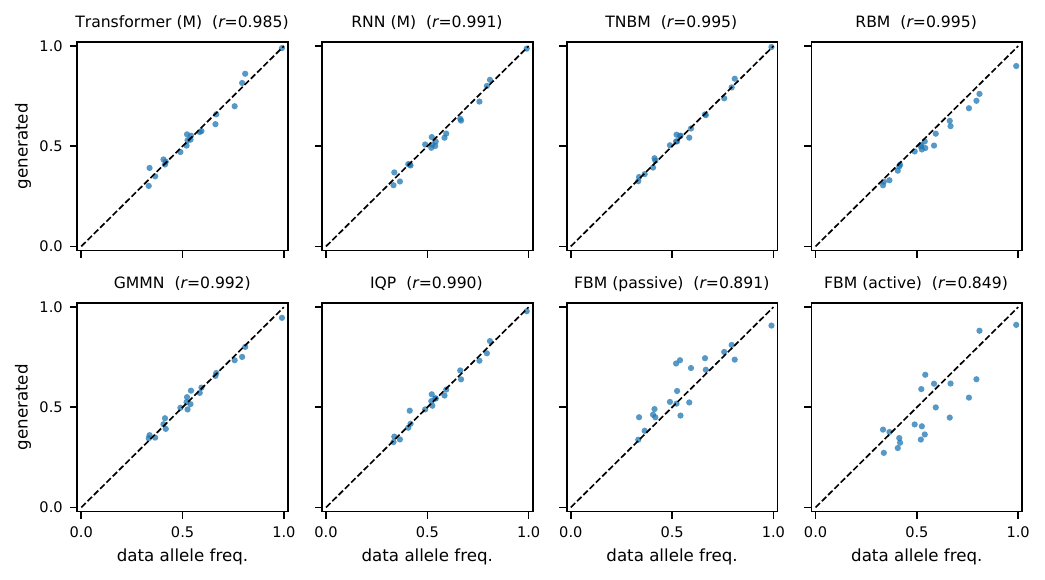}
\caption{\textbf{Allele frequencies.} Per-locus allele frequencies of generated samples against the reference data
($\varepsilon{=}20\%$, $N{=}20$), for eight of the benchmarked models; $r$ is the Pearson correlation with the reference frequencies.}
\label{fig:af_ld}
\end{figure}

\begin{figure}[!ht]
\centering
\includegraphics[width=0.78\columnwidth]{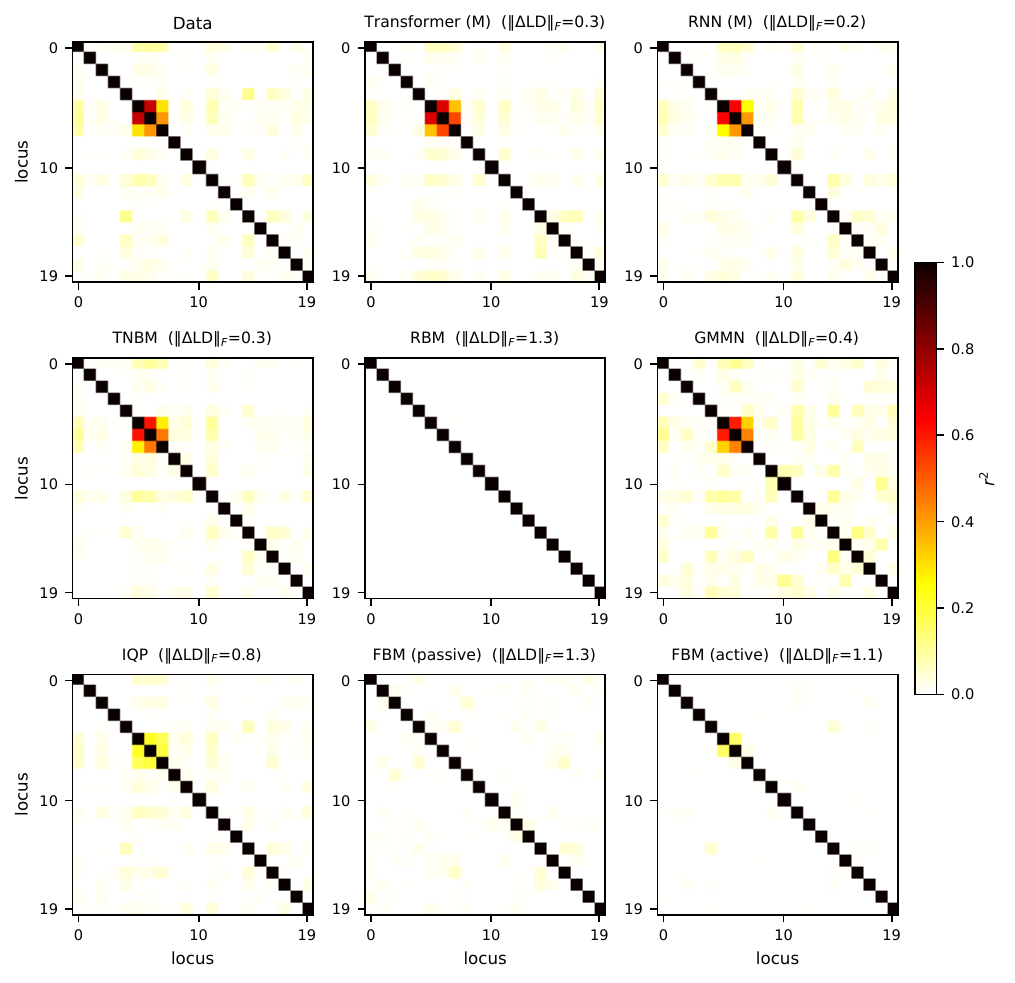}
\caption{\textbf{Linkage disequilibrium.} Pairwise linkage-disequilibrium ($r^2$) matrices for the reference data and eight of the
benchmarked models at $\varepsilon{=}20\%$, $N{=}20$; each panel reports the Frobenius distance to the reference matrix.}
\label{fig:ldmat}
\end{figure}

\begin{figure}[!ht]
\centering
\includegraphics[width=0.95\columnwidth]{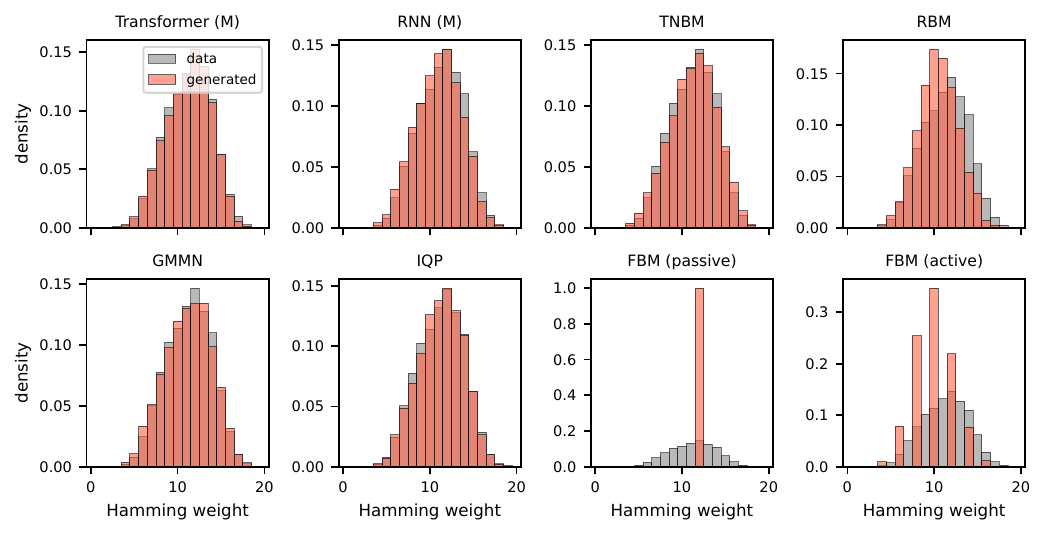}
\caption{\textbf{Hamming-weight distributions.} Hamming-weight distributions of generated samples (red) against the reference data
(gray) at $\varepsilon{=}20\%$, $N{=}20$.}
\label{fig:hwdist}
\end{figure}

\begin{figure}[!ht]
\centering
\includegraphics[width=0.95\columnwidth]{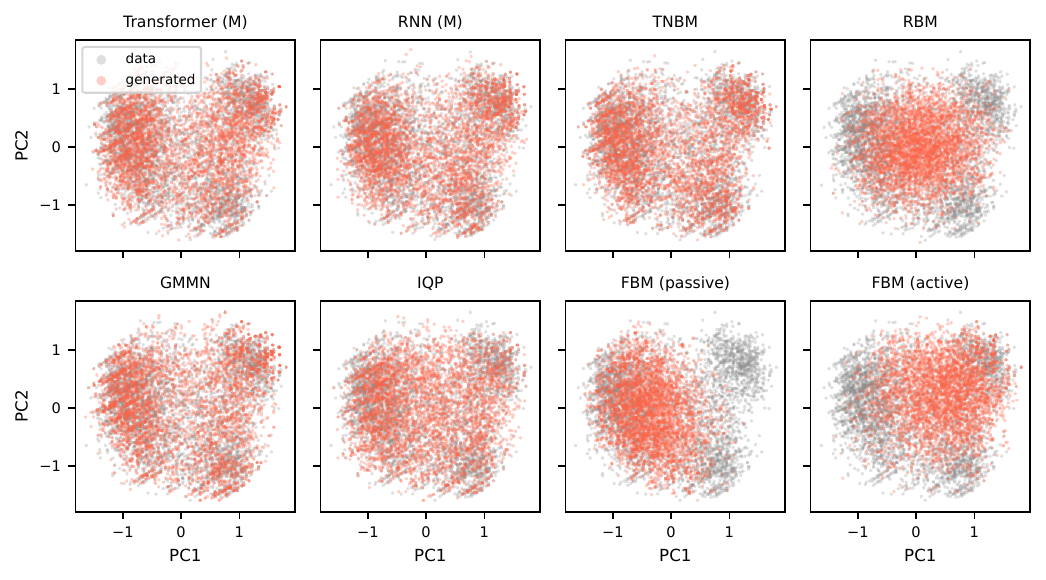}
\caption{\textbf{PCA projections.} Projections of generated samples (red) and reference data (gray) onto the first two principal
components of the reference data, at $\varepsilon{=}20\%$, $N{=}20$.}
\label{fig:genpca}
\end{figure}

\end{document}

%% file: tables/master_combined.tex
\begin{tabular}{l l *{3}{c} @{\hspace{1.5em}}!{\vrule} @{\hspace{1.5em}} *{2}{c}}
\toprule
& & \multicolumn{3}{c}{\textbf{Cardinality} ($N{=}30$)} & \multicolumn{2}{c}{\textbf{Genomic} ($N{=}20$)}\\
\cmidrule(lr){3-5}\cmidrule(lr){6-7}
Model & loss & MMD$^2$ & $F$ & $\tilde{C}$ & MMD$^2$ & $C$\\
\midrule
\multicolumn{7}{l}{\emph{Likelihood-trained}}\\
Transformer (L) & cross-ent & \cellcolor[rgb]{0.62,0.86,0.58} 1.00e{-}3 & 1.00 & \cellcolor[rgb]{0.62,0.86,0.58} 1.00 & \cellcolor[rgb]{0.73,0.90,0.70} 1.58e{-}3 & \cellcolor[rgb]{0.68,0.88,0.65} 0.130\\
Transformer (M) & cross-ent & \cellcolor[rgb]{0.62,0.86,0.58} 9.99e{-}4 & 1.00 & \cellcolor[rgb]{0.62,0.86,0.58} 1.00 & \cellcolor[rgb]{0.72,0.90,0.69} 1.49e{-}3 & \cellcolor[rgb]{0.63,0.86,0.59} 0.140\\
RNN (S) & cross-ent & \cellcolor[rgb]{0.63,0.86,0.59} 1.02e{-}3 & 1.00 & \cellcolor[rgb]{0.62,0.86,0.58} 1.00 & \cellcolor[rgb]{0.87,0.95,0.86} 3.05e{-}3 & \cellcolor[rgb]{0.93,0.97,0.92} 0.084\\
TNBM & NLL & \cellcolor[rgb]{0.63,0.86,0.59} 1.01e{-}3 & 0.96 & \cellcolor[rgb]{0.65,0.87,0.62} 0.96 & \cellcolor[rgb]{0.62,0.86,0.58} 8.17e{-}4 & \cellcolor[rgb]{0.63,0.86,0.59} 0.140\\
RNN (M) & cross-ent & \cellcolor[rgb]{0.63,0.86,0.59} 1.01e{-}3 & 0.91 & \cellcolor[rgb]{0.68,0.88,0.65} 0.91 & \cellcolor[rgb]{0.66,0.88,0.63} 1.08e{-}3 & \cellcolor[rgb]{0.62,0.86,0.58} 0.141\\
RNN (L) & cross-ent & \cellcolor[rgb]{0.63,0.86,0.59} 1.01e{-}3 & 0.88 & \cellcolor[rgb]{0.71,0.89,0.68} 0.88 & \cellcolor[rgb]{0.68,0.88,0.65} 1.21e{-}3 & \cellcolor[rgb]{0.67,0.88,0.64} 0.132\\
Transformer (S) & cross-ent & \cellcolor[rgb]{0.62,0.86,0.58} 9.97e{-}4 & 0.70 & \cellcolor[rgb]{0.85,0.94,0.83} 0.70 & \cellcolor[rgb]{0.86,0.95,0.85} 2.93e{-}3 & \cellcolor[rgb]{0.91,0.97,0.90} 0.087\\
RBM & CD & \cellcolor[rgb]{0.64,0.87,0.61} 1.04e{-}3 & 0.15 & \cellcolor[rgb]{0.96,0.68,0.68} 0.15 & \cellcolor[rgb]{0.98,0.85,0.85} 5.62e{-}3 & \cellcolor[rgb]{0.96,0.71,0.71} 0.025\\
\midrule
\multicolumn{7}{l}{\emph{Moment-trained}}\\
FBM (passive) & Z-corr & \cellcolor[rgb]{0.64,0.87,0.60} 1.03e{-}3 & 1.00 & \cellcolor[rgb]{0.62,0.86,0.58} 1.00 & \cellcolor[rgb]{0.94,0.55,0.55} 6.43e{-}3 & \cellcolor[rgb]{0.96,0.74,0.74} 0.029\\
magic FBM$^{\ddagger}$ & MMD$^2$ & \cellcolor[rgb]{0.66,0.87,0.62} 1.06e{-}3 & 0.31 & \cellcolor[rgb]{0.98,0.83,0.83} 0.31 & \cellcolor[rgb]{0.94,0.98,0.93} 3.80e{-}3 & \cellcolor[rgb]{0.98,0.84,0.84} 0.046\\
FBM (active)$^{\ddagger}$ & Z-corr & \cellcolor[rgb]{0.65,0.87,0.61} 1.05e{-}3 & 0.38 & \cellcolor[rgb]{0.99,0.89,0.89} 0.38 & \cellcolor[rgb]{0.97,0.77,0.77} 6.04e{-}3 & \cellcolor[rgb]{0.98,0.88,0.88} 0.052\\
GMMN (classical MMD) & MMD$^2$ & \cellcolor[rgb]{0.94,0.55,0.55} 1.73e{-}3 & 0.24 & \cellcolor[rgb]{0.97,0.76,0.76} 0.24 & \cellcolor[rgb]{0.78,0.92,0.76} 2.07e{-}3 & \cellcolor[rgb]{0.69,0.88,0.65} 0.129\\
IQP & MMD$^2$ & \cellcolor[rgb]{0.88,0.96,0.87} 1.43e{-}3 & 0.14 & \cellcolor[rgb]{0.96,0.68,0.68} 0.14 & \cellcolor[rgb]{0.73,0.90,0.70} 1.57e{-}3 & \cellcolor[rgb]{0.77,0.92,0.75} 0.112\\
\bottomrule
\end{tabular}

%% file: tables/master_n16.tex
\begin{table}\centering\scriptsize
\caption{\textbf{Cardinality benchmark at $N=16$, HW$=8$, $\epsilon=0.01$.} $Q=10,000$, raw sampling, mean over seeds. MMD$^2$ is the eval Gaussian--Hamming discrepancy to held-out test (computed post-training for every model). IQP uses the marginal-matching initialization (Appendix~\ref{app:init}). The magic FBM uses the all-measured configuration (Appendix~\ref{app:magic}).}
\begin{tabular}{l l r r r | r r r r}
\toprule
Model & training loss & params & MMD$^2_{\sigma=1}$ & MMD$^2_{\rm med}$ & $E$ & $F$ & $R$ & $C$\\
\midrule
TNBM & NLL (DMRG, cumulant) & 1132 & 3.26e{-}3 & 4.38e{-}4 & 0.960 & 0.258 & 0.247 & \textbf{0.116}\\
FBM (passive) & Z-string correlator MSE & 120 & 3.04e{-}3 & 5.38e{-}4 & 0.975 & 1.000 & 0.975 & \textbf{0.377}\\
FBM (active) & Z-string correlator MSE & 496 & 2.95e{-}3 & 2.73e{-}3 & 0.991 & 0.429 & 0.425 & \textbf{0.245}\\
magic FBM & MMD$^2$ (stochastic) & 500 & 2.80e{-}3 & 2.44e{-}3 & 0.994 & 0.471 & 0.468 & \textbf{0.228}\\
IQP & MMD$^2$ (den-Nest) & 136 & 3.50e{-}3 & 2.88e{-}3 & 0.995 & 0.252 & 0.251 & \textbf{0.156}\\
Transformer (S) & cross-entropy (NLL) & 3394 & 1.54e{-}3 & 2.80e{-}4 & 0.992 & 0.650 & 0.644 & \textbf{0.380}\\
Transformer (M) & cross-entropy (NLL) & 25634 & 2.11e{-}3 & 4.62e{-}4 & 0.989 & 0.893 & 0.883 & \textbf{0.461}\\
Transformer (L) & cross-entropy (NLL) & 150402 & 2.02e{-}3 & 4.13e{-}4 & 0.988 & 0.969 & 0.957 & \textbf{0.498}\\
RNN (S) & cross-entropy (NLL) & 1714 & 1.33e{-}3 & 1.76e{-}4 & 0.991 & 0.727 & 0.721 & \textbf{0.416}\\
RNN (M) & cross-entropy (NLL) & 25282 & 1.51e{-}3 & 2.33e{-}4 & 0.989 & 0.913 & 0.903 & \textbf{0.479}\\
RNN (L) & cross-entropy (NLL) & 198786 & 1.83e{-}3 & 2.94e{-}4 & 0.982 & 0.971 & 0.953 & \textbf{0.453}\\
RBM & contrastive divergence & 424 & 2.68e{-}3 & 4.46e{-}4 & 0.998 & 0.193 & 0.193 & \textbf{0.139}\\
\bottomrule
\end{tabular}\end{table}

\begin{table}\centering\scriptsize
\caption{\textbf{Cardinality benchmark at $N=16$, HW$=8$, $\epsilon=0.05$.} $Q=10,000$, raw sampling, mean over seeds. MMD$^2$ is the eval Gaussian--Hamming discrepancy to held-out test (computed post-training for every model). IQP uses the marginal-matching initialization (Appendix~\ref{app:init}). The magic FBM uses the all-measured configuration (Appendix~\ref{app:magic}).}
\begin{tabular}{l l r r r | r r r r}
\toprule
Model & training loss & params & MMD$^2_{\sigma=1}$ & MMD$^2_{\rm med}$ & $E$ & $F$ & $R$ & $C$\\
\midrule
TNBM & NLL (DMRG, cumulant) & 1132 & 1.35e{-}3 & 1.79e{-}4 & 0.937 & 0.901 & 0.844 & \textbf{0.470}\\
FBM (passive) & Z-string correlator MSE & 120 & 2.17e{-}3 & 2.20e{-}4 & 0.932 & 1.000 & 0.932 & \textbf{0.388}\\
FBM (active) & Z-string correlator MSE & 496 & 2.40e{-}3 & 1.40e{-}3 & 0.976 & 0.402 & 0.393 & \textbf{0.242}\\
magic FBM & MMD$^2$ (stochastic) & 500 & 2.03e{-}3 & 1.03e{-}3 & 0.974 & 0.462 & 0.450 & \textbf{0.242}\\
IQP & MMD$^2$ (den-Nest) & 136 & 2.09e{-}3 & 1.24e{-}3 & 0.985 & 0.246 & 0.242 & \textbf{0.167}\\
Transformer (S) & cross-entropy (NLL) & 3394 & 1.12e{-}3 & 1.50e{-}4 & 0.962 & 0.750 & 0.721 & \textbf{0.440}\\
Transformer (M) & cross-entropy (NLL) & 25634 & 1.92e{-}3 & 4.86e{-}4 & 0.959 & 0.815 & 0.781 & \textbf{0.455}\\
Transformer (L) & cross-entropy (NLL) & 150402 & 1.35e{-}3 & 2.34e{-}4 & 0.949 & 0.989 & 0.939 & \textbf{0.524}\\
RNN (S) & cross-entropy (NLL) & 1714 & 1.03e{-}3 & 9.06e{-}5 & 0.953 & 0.918 & 0.875 & \textbf{0.505}\\
RNN (M) & cross-entropy (NLL) & 25282 & 1.65e{-}3 & 3.37e{-}4 & 0.949 & 0.931 & 0.884 & \textbf{0.490}\\
RNN (L) & cross-entropy (NLL) & 198786 & 1.29e{-}3 & 1.93e{-}4 & 0.947 & 0.983 & 0.931 & \textbf{0.515}\\
RBM & contrastive divergence & 424 & 1.36e{-}1 & 2.56e{-}2 & 1.000 & 0.003 & 0.003 & \textbf{0.002}\\
\bottomrule
\end{tabular}\end{table}

\begin{table}\centering\scriptsize
\caption{\textbf{Cardinality benchmark at $N=16$, HW$=8$, $\epsilon=0.10$.} $Q=10,000$, raw sampling, mean over seeds. MMD$^2$ is the eval Gaussian--Hamming discrepancy to held-out test (computed post-training for every model). IQP uses the marginal-matching initialization (Appendix~\ref{app:init}). The magic FBM uses the all-measured configuration (Appendix~\ref{app:magic}).}
\begin{tabular}{l l r r r | r r r r}
\toprule
Model & training loss & params & MMD$^2_{\sigma=1}$ & MMD$^2_{\rm med}$ & $E$ & $F$ & $R$ & $C$\\
\midrule
TNBM & NLL (DMRG, cumulant) & 1088 & 1.12e{-}3 & 1.17e{-}4 & 0.890 & 0.971 & 0.864 & \textbf{0.515}\\
FBM (passive) & Z-string correlator MSE & 120 & 1.86e{-}3 & 1.22e{-}4 & 0.876 & 1.000 & 0.876 & \textbf{0.395}\\
FBM (active) & Z-string correlator MSE & 496 & 2.28e{-}3 & 9.09e{-}4 & 0.955 & 0.384 & 0.366 & \textbf{0.240}\\
magic FBM & MMD$^2$ (stochastic) & 500 & 1.91e{-}3 & 9.19e{-}4 & 0.953 & 0.448 & 0.427 & \textbf{0.253}\\
IQP & MMD$^2$ (den-Nest) & 136 & 1.91e{-}3 & 9.08e{-}4 & 0.971 & 0.241 & 0.234 & \textbf{0.171}\\
Transformer (S) & cross-entropy (NLL) & 3394 & 1.00e{-}3 & 7.41e{-}5 & 0.911 & 0.829 & 0.755 & \textbf{0.477}\\
Transformer (M) & cross-entropy (NLL) & 25634 & 1.31e{-}3 & 2.37e{-}4 & 0.903 & 0.971 & 0.877 & \textbf{0.523}\\
Transformer (L) & cross-entropy (NLL) & 150402 & 1.54e{-}3 & 3.41e{-}4 & 0.902 & 0.978 & 0.882 & \textbf{0.522}\\
RNN (S) & cross-entropy (NLL) & 1714 & 9.85e{-}4 & 7.39e{-}5 & 0.900 & 0.965 & 0.868 & \textbf{0.524}\\
RNN (M) & cross-entropy (NLL) & 25282 & 1.05e{-}3 & 9.52e{-}5 & 0.897 & 0.978 & 0.877 & \textbf{0.524}\\
RNN (L) & cross-entropy (NLL) & 198786 & 1.17e{-}3 & 1.69e{-}4 & 0.903 & 0.946 & 0.854 & \textbf{0.511}\\
RBM & contrastive divergence & 424 & 1.63e{-}2 & 4.94e{-}3 & 0.989 & 0.097 & 0.096 & \textbf{0.080}\\
\bottomrule
\end{tabular}\end{table}

%% file: tables/master_n20.tex
\begin{table}\centering\scriptsize
\caption{\textbf{Cardinality benchmark at $N=20$, HW$=10$, $\epsilon=0.01$.} $Q=100,000$, raw sampling, mean over seeds. MMD$^2$ is the eval Gaussian--Hamming discrepancy to held-out test (computed post-training for every model). IQP uses the marginal-matching initialization (Appendix~\ref{app:init}). The magic FBM uses the all-measured configuration (Appendix~\ref{app:magic}).}
\begin{tabular}{l l r r r | r r r r}
\toprule
Model & training loss & params & MMD$^2_{\sigma=1}$ & MMD$^2_{\rm med}$ & $E$ & $F$ & $R$ & $C$\\
\midrule
TNBM & NLL (DMRG, cumulant) & 1422 & 1.04e{-}3 & 8.97e{-}5 & 0.990 & 0.986 & 0.976 & \textbf{0.408}\\
FBM (passive) & Z-string correlator MSE & 190 & 1.27e{-}3 & 7.74e{-}5 & 0.988 & 1.000 & 0.988 & \textbf{0.305}\\
FBM (active) & Z-string correlator MSE & 780 & 1.60e{-}3 & 1.00e{-}3 & 0.996 & 0.358 & 0.357 & \textbf{0.160}\\
magic FBM & MMD$^2$ (stochastic) & 785 & 1.34e{-}3 & 6.38e{-}4 & 0.996 & 0.388 & 0.387 & \textbf{0.169}\\
IQP & MMD$^2$ (den-Nest) & 210 & 1.39e{-}3 & 7.00e{-}4 & 0.998 & 0.186 & 0.185 & \textbf{0.092}\\
Transformer (S) & cross-entropy (NLL) & 3394 & 1.03e{-}3 & 7.00e{-}5 & 0.992 & 0.789 & 0.783 & \textbf{0.346}\\
Transformer (M) & cross-entropy (NLL) & 25634 & 1.10e{-}3 & 1.36e{-}4 & 0.990 & 0.989 & 0.979 & \textbf{0.411}\\
Transformer (L) & cross-entropy (NLL) & 150402 & 1.09e{-}3 & 1.27e{-}4 & 0.990 & 0.994 & 0.983 & \textbf{0.411}\\
RNN (S) & cross-entropy (NLL) & 1714 & 1.09e{-}3 & 1.30e{-}4 & 0.990 & 0.976 & 0.966 & \textbf{0.405}\\
RNN (M) & cross-entropy (NLL) & 25282 & 1.08e{-}3 & 1.02e{-}4 & 0.990 & 0.989 & 0.979 & \textbf{0.410}\\
RNN (L) & cross-entropy (NLL) & 198786 & 1.36e{-}3 & 2.38e{-}4 & 0.990 & 0.967 & 0.957 & \textbf{0.389}\\
RBM & contrastive divergence & 524 & 1.43e{-}3 & 8.19e{-}5 & 0.998 & 0.175 & 0.175 & \textbf{0.091}\\
\bottomrule
\end{tabular}\end{table}

\begin{table}\centering\scriptsize
\caption{\textbf{Cardinality benchmark at $N=20$, HW$=10$, $\epsilon=0.05$.} $Q=100,000$, raw sampling, mean over seeds. MMD$^2$ is the eval Gaussian--Hamming discrepancy to held-out test (computed post-training for every model). IQP uses the marginal-matching initialization (Appendix~\ref{app:init}). The magic FBM uses the all-measured configuration (Appendix~\ref{app:magic}).}
\begin{tabular}{l l r r r | r r r r}
\toprule
Model & training loss & params & MMD$^2_{\sigma=1}$ & MMD$^2_{\rm med}$ & $E$ & $F$ & $R$ & $C$\\
\midrule
TNBM & NLL (DMRG, cumulant) & 1244 & 9.90e{-}4 & 5.71e{-}5 & 0.949 & 0.993 & 0.943 & \textbf{0.414}\\
FBM (passive) & Z-string correlator MSE & 190 & 1.31e{-}3 & 6.09e{-}5 & 0.947 & 1.000 & 0.947 & \textbf{0.308}\\
FBM (active) & Z-string correlator MSE & 780 & 1.56e{-}3 & 8.53e{-}4 & 0.982 & 0.345 & 0.339 & \textbf{0.159}\\
magic FBM & MMD$^2$ (stochastic) & 785 & 1.38e{-}3 & 6.20e{-}4 & 0.980 & 0.375 & 0.367 & \textbf{0.171}\\
IQP & MMD$^2$ (den-Nest) & 210 & 1.36e{-}3 & 5.99e{-}4 & 0.990 & 0.178 & 0.176 & \textbf{0.091}\\
Transformer (S) & cross-entropy (NLL) & 3394 & 1.01e{-}3 & 5.54e{-}5 & 0.960 & 0.803 & 0.770 & \textbf{0.354}\\
Transformer (M) & cross-entropy (NLL) & 25634 & 1.02e{-}3 & 6.93e{-}5 & 0.949 & 0.997 & 0.947 & \textbf{0.416}\\
Transformer (L) & cross-entropy (NLL) & 150402 & 1.16e{-}3 & 1.72e{-}4 & 0.950 & 0.998 & 0.948 & \textbf{0.413}\\
RNN (S) & cross-entropy (NLL) & 1714 & 1.06e{-}3 & 9.39e{-}5 & 0.949 & 0.998 & 0.948 & \textbf{0.416}\\
RNN (M) & cross-entropy (NLL) & 25282 & 1.09e{-}3 & 1.26e{-}4 & 0.953 & 0.952 & 0.907 & \textbf{0.400}\\
RNN (L) & cross-entropy (NLL) & 198786 & 1.11e{-}3 & 1.01e{-}4 & 0.953 & 0.929 & 0.885 & \textbf{0.389}\\
RBM & contrastive divergence & 524 & 1.44e{-}3 & 8.79e{-}5 & 0.991 & 0.168 & 0.167 & \textbf{0.090}\\
\bottomrule
\end{tabular}\end{table}

\begin{table}\centering\scriptsize
\caption{\textbf{Cardinality benchmark at $N=20$, HW$=10$, $\epsilon=0.10$.} $Q=100,000$, raw sampling, mean over seeds. MMD$^2$ is the eval Gaussian--Hamming discrepancy to held-out test (computed post-training for every model). IQP uses the marginal-matching initialization (Appendix~\ref{app:init}). The magic FBM uses the all-measured configuration (Appendix~\ref{app:magic}).}
\begin{tabular}{l l r r r | r r r r}
\toprule
Model & training loss & params & MMD$^2_{\sigma=1}$ & MMD$^2_{\rm med}$ & $E$ & $F$ & $R$ & $C$\\
\midrule
TNBM & NLL (DMRG, cumulant) & 1232 & 1.01e{-}3 & 5.35e{-}5 & 0.899 & 0.992 & 0.893 & \textbf{0.414}\\
FBM (passive) & Z-string correlator MSE & 190 & 1.34e{-}3 & 4.63e{-}5 & 0.895 & 1.000 & 0.895 & \textbf{0.308}\\
FBM (active) & Z-string correlator MSE & 780 & 1.59e{-}3 & 7.91e{-}4 & 0.963 & 0.333 & 0.320 & \textbf{0.159}\\
magic FBM & MMD$^2$ (stochastic) & 785 & 1.36e{-}3 & 5.25e{-}4 & 0.961 & 0.365 & 0.350 & \textbf{0.168}\\
IQP & MMD$^2$ (den-Nest) & 210 & 1.42e{-}3 & 8.13e{-}4 & 0.981 & 0.173 & 0.170 & \textbf{0.093}\\
Transformer (S) & cross-entropy (NLL) & 3394 & 1.03e{-}3 & 6.13e{-}5 & 0.915 & 0.829 & 0.759 & \textbf{0.365}\\
Transformer (M) & cross-entropy (NLL) & 25634 & 1.02e{-}3 & 5.79e{-}5 & 0.900 & 0.999 & 0.899 & \textbf{0.416}\\
Transformer (L) & cross-entropy (NLL) & 150402 & 9.82e{-}4 & 4.25e{-}5 & 0.900 & 0.999 & 0.900 & \textbf{0.416}\\
RNN (S) & cross-entropy (NLL) & 1714 & 1.00e{-}3 & 5.59e{-}5 & 0.901 & 0.995 & 0.896 & \textbf{0.416}\\
RNN (M) & cross-entropy (NLL) & 25282 & 1.07e{-}3 & 1.10e{-}4 & 0.905 & 0.945 & 0.855 & \textbf{0.398}\\
RNN (L) & cross-entropy (NLL) & 198786 & 1.12e{-}3 & 1.41e{-}4 & 0.908 & 0.911 & 0.827 & \textbf{0.387}\\
RBM & contrastive divergence & 524 & 1.41e{-}3 & 7.52e{-}5 & 0.982 & 0.163 & 0.160 & \textbf{0.092}\\
\bottomrule
\end{tabular}\end{table}